\pdfoutput=1
\documentclass[acmsmall,screen]{acmart}

\setcopyright{rightsretained}
\acmDOI{10.1145/3649811}
\acmYear{2024}
\copyrightyear{2024}
\acmSubmissionID{oopslaa24main-p8-p}
\acmJournal{PACMPL}
\acmVolume{8}
\acmNumber{OOPSLA1}
\acmArticle{94}
\acmMonth{4}
\received{16-OCT-2023}
\received[accepted]{2024-02-24}



\usepackage{braket}
\usepackage{cleveref}
\usepackage{enumitem}
\usepackage{listings}
\usepackage{mathtools}
\usepackage{soul}
\usepackage{tikz}
\usepackage{wrapfig}

\usetikzlibrary{decorations.pathmorphing}

\newlength\LineWidth
\newlength\Amplitude
\newlength\SegLength

\definecolor{HLcolor}{RGB}{240,0,0}

\newcommand\tikzmark[1]{%
  \tikz[overlay,remember picture] \node (#1) {};}

\makeatletter
\newcommand{\highlight@DoHighlight}{
  \draw[HLcolor,line width=\LineWidth,decorate,decoration={zigzag,amplitude=\Amplitude,segment length=\SegLength}]  ($(begin highlight)+(0,-2pt)$) -- ($(end highlight)+(0,-2pt)$) ;
}

\newcommand{\highlight@BeginHighlight}{
  \coordinate (begin highlight) at (0,0) ;
}

\newcommand{\highlight@EndHighlight}{
  \coordinate (end highlight) at (0,0) ;
}

\newdimen\highlight@previous
\newdimen\highlight@current

\DeclareRobustCommand*\highlight[1][]{%
  \SOUL@setup
  \def\SOUL@preamble{%
    \begin{tikzpicture}[overlay, remember picture]
      \highlight@BeginHighlight
      \highlight@EndHighlight
    \end{tikzpicture}%
  }%
  \def\SOUL@postamble{%
    \begin{tikzpicture}[overlay, remember picture]
      \highlight@EndHighlight
      \highlight@DoHighlight
    \end{tikzpicture}%
  }%
  \def\SOUL@everyhyphen{%
    \discretionary{%
      \SOUL@setkern\SOUL@hyphkern
      \SOUL@sethyphenchar
      \tikz[overlay, remember picture] \highlight@EndHighlight ;%
    }{%
    }{%
      \SOUL@setkern\SOUL@charkern
    }%
  }%
  \def\SOUL@everyexhyphen##1{%
    \SOUL@setkern\SOUL@hyphkern
    \hbox{##1}%
    \discretionary{%
      \tikz[overlay, remember picture] \highlight@EndHighlight ;%
    }{%
    }{%
      \SOUL@setkern\SOUL@charkern
    }%
  }%
  \def\SOUL@everysyllable{%
    \begin{tikzpicture}[overlay, remember picture]
      \path let \p0 = (begin highlight), \p1 = (0,0) in \pgfextra
        \global\highlight@previous=\y0
        \global\highlight@current =\y1
      \endpgfextra (0,0) ;
      \ifdim\highlight@current < \highlight@previous
        \highlight@DoHighlight
        \highlight@BeginHighlight
      \fi
    \end{tikzpicture}%
    \the\SOUL@syllable
    \tikz[overlay, remember picture] \highlight@EndHighlight ;%
  }%
  \SOUL@
}
\makeatother

\DeclareDocumentCommand\MarkText{O{red}O{0.75pt}O{5pt}m}{%
  \colorlet{HLcolor}{#1}
  \setlength\Amplitude{#2}%
  \setlength\SegLength{#3}%
  \tikzmark{endquote}\tikzmark{beginquote}\highlight{#4}%
}

\usetikzlibrary{calc}

\definecolor{mygray}{gray}{0.4}
\newcommand{\cancel}[1]{%
    \tikz[baseline=(tocancel.base)]{
        \node[inner sep=0pt,outer sep=0pt] (tocancel) {\ensuremath{#1}};
        \draw[mygray] ($(tocancel.south west)$) -- ($(tocancel.north east)$);
    }%
}%

\DeclareRobustCommand{\minwidthbox}[2]{%
  \ifmmode
    \expandafter\mathmakebox
  \else
    \expandafter\makebox
  \fi
  [\ifdim#2<\width\width\else#2\fi]{#1}%
}

\newcommand{\mymapsto}[2]{\hphantom{\mapsto}\smash{\llap{\ensuremath{\xmapsto{\minwidthbox{#1}{#2}}{}}}}}

\crefname{lemma}{Lemma}{Lemmas}
\Crefname{lemma}{Lemma}{Lemmas}

\newcommand{\Mod}[1]{\ \mathrm{mod}\ #1}
\DeclarePairedDelimiter\abs{\lvert}{\rvert}%
\DeclarePairedDelimiter\norm{\lVert}{\rVert}%

\begin{document}

\author{Charles Yuan}
\affiliation{
  \institution{MIT CSAIL}
  \streetaddress{32 Vassar St}
  \city{Cambridge}
  \state{MA}
  \postcode{02139}
  \country{USA}
}
\email{chenhuiy@csail.mit.edu}
\orcid{0000-0002-4918-4467}

\author{Agnes Villanyi}
\affiliation{
  \institution{MIT CSAIL}
  \streetaddress{32 Vassar St}
  \city{Cambridge}
  \state{MA}
  \postcode{02139}
  \country{USA}
}
\email{agivilla@mit.edu}
\orcid{0009-0005-3121-2537}

\author{Michael Carbin}
\affiliation{
  \institution{MIT CSAIL}
  \streetaddress{32 Vassar St}
  \city{Cambridge}
  \state{MA}
  \postcode{02139}
  \country{USA}
}
\email{mcarbin@csail.mit.edu}
\orcid{0000-0002-6928-0456}

\title[Quantum Control Machine: The Limits of Control Flow in Quantum Programming]{Quantum Control Machine}
\subtitle{The Limits of Control Flow in Quantum Programming}

\begin{abstract}
Quantum algorithms for tasks such as factorization, search, and simulation rely on control flow such as branching and iteration that depends on the value of data in superposition. High-level programming abstractions for control flow, such as switches, loops, higher-order functions, and continuations, are ubiquitous in classical languages. By contrast, many quantum languages do not provide high-level abstractions for control flow in superposition, and instead require the use of hardware-level logic gates to implement such control flow.

The reason for this gap is that whereas a classical computer supports control flow abstractions using a program counter that can depend on data, the typical architecture of a quantum computer does not analogously provide a program counter that can depend on data in superposition.
As a result, the complete set of control flow abstractions that can be correctly realized on a quantum computer has not yet been established.

In this work, we provide a complete characterization of the properties of control flow abstractions that are correctly realizable on a quantum computer. First, we prove that even on a quantum computer whose program counter exists in superposition, one cannot correctly realize control flow in quantum algorithms by lifting the classical conditional jump instruction to work in superposition. This theorem denies the ability to directly lift general abstractions for control flow such as the $\lambda$-calculus from classical to quantum programming.

In response, we present the necessary and sufficient conditions for control flow to be correctly realizable on a quantum computer. We introduce the quantum control machine, an instruction set architecture featuring a conditional jump that is restricted to satisfy these conditions. We show how this design enables a developer to correctly express control flow in quantum algorithms using a program counter in place of logic gates.
\end{abstract}

\maketitle

\section{Introduction}

\emph{Quantum algorithms} promise computational advantage in areas ranging from factorization~\citep{shor1997} and search~\citep{grover1996} to data analysis~\citep{harrow2009,wiebe2012,lloyd2015} and simulation~\citep{babbush2018,childs2018,abrams1997}.

The power of a quantum algorithm arises from its ability to manipulate quantum data, which exists in a weighted sum over many classical states known as a \emph{superposition}. The basic unit of quantum information is the \emph{qubit} --- a superposition of the bits 0 and 1. A quantum computer can transform the values and weights within a superposition by performing quantum logic gates, formally known as \emph{unitary operators}. It can also \emph{measure} quantum data, which collapses a superposition to a classical state with a probability that is determined by its weight in the superposition.

One example of data in superposition is a 2-qubit quantum integer $k$ that takes on the values 0, 1, 2, and 3 at the same time. In the notation of quantum algorithms, one such $k$ is written as $\frac{1}{2} \ket{0} + \frac{1}{2} \ket{1} - \frac{1}{2} \ket{2} + \frac{1}{2} \ket{3}$ where $\ket{2}$ or $\ket{10_2}$ in binary denotes a qubit in the 1 state and a qubit in the 0 state, and $- \frac{1}{2}$ denotes the weight of that state, whose sign component is called a \emph{phase}.

Unlike classical probabilities that are non-negative real numbers, the weights in a superposition are complex numbers whose values can combine or cancel when added, in a phenomenon known as \emph{interference}. In turn, interference enables quantum advantage by amplifying the probability that a quantum algorithm produces a correct result to be larger than that of any classical algorithm.

A widely adopted representation of a quantum computation is as a \emph{quantum circuit}, a sequence of quantum logic gates operating over qubits that is the quantum analogue of a Boolean logic circuit.
To assist in manipulating the complex quantum circuits that arise when implementing quantum algorithms, researchers have developed \emph{quantum programming languages}~\citep{qsharp, quipper, qwire, qml, silq, voichick2023}.

\subsection{The Challenge of Control Flow in Superposition}

A programming abstraction that is integral to classical algorithms is control flow such as branching and iteration that depends on the value of data.
An analogous concept is also integral to quantum algorithms, in which control flow depends on the value of data in quantum superposition.

\begin{example}[Branching] \label{ex:branching}
  Where a classical computation takes a value $k$ and executes the $k$th branch of a switch statement on the value $x$,
  a quantum computation takes a superposition of $k$ and executes a superposition of the corresponding branches to produce a superposition of $x$.

  In the formal notation of quantum algorithms, given a set of functions $\{U_i\}$ representing the branches, branching transforms the data from $\sum_k \ket{k} \ket{x} \mapsto \sum_k \ket{k} (U_k \ket{x})$.
  This operation is used in algorithms for physical simulation~\citep{babbush2018,childs2012,low2019}, in which each $U_i$ encodes a component of the description of the target system.
\end{example}

\begin{example}[Iteration] \label{ex:iteration}
  Where a classical computation takes a value $k$ and repeats an operation for $k$ iterations on the value $x$, a quantum computation takes a superposition of $k$ and repeats the operation for a superposition of numbers of iterations to produce a superposition of $x$.

  In formal notation, given a function $U$ whose $i$th power is $U^i$, iteration transforms the data from $\sum_k \ket{k} \ket{x} \mapsto \sum_k \ket{k} (U^k \ket{x})$.
  Iteration is a special case of branching that is used in algorithms for factoring~\citep{shor1997}, where $U$ maps $\ket{x} \mapsto \ket{ax \Mod N}$, and for phase estimation~\citep{kitaev1995} as found in simulation~\citep{abrams1997} and linear algebra~\citep{harrow2009}.
\end{example}

Programming abstractions for control flow are natively supported by the typical architecture of a classical computer.
In imperative programming, the \texttt{if}-statement for branching and \texttt{for}-loop for iteration compile to a \emph{program counter} that determines the current instruction and a \emph{conditional jump} instruction that updates the program counter using a condition on a data register.
Control flow is also straightforward to realize in functional programming, in which abstractions for branching and iteration emerge from the Church encoding~\citep{church1941} of data types in the $\lambda$-calculus.

By contrast, programming abstractions for control flow in superposition are not natively supported by the typical architecture of a quantum computer.
Whereas a classical computer provides a program counter that can depend on data, the typical architecture of a quantum computer does not provide a program counter that can depend on data in superposition, nor a representation of $\lambda$-terms in superposition. Instead, it requires a program to be represented as a quantum circuit, a fixed sequence of bit-level logic gates whose structure cannot depend on data in superposition.

In turn, the lack of a program counter requires alternatives to the typical strategy of compiling abstractions for control flow to manipulations of the program counter.
For example, some quantum languages~\citep{quipper,qwire,qsharp,qiskit} enable the developer to implement \Cref{ex:branching,ex:iteration} by explicitly building a circuit of bit-controlled logic gates.
Alternatively, emerging languages~\citep{silq,tower} provide certain abstractions for control flow, such as a quantum \texttt{if}-statement that branches over a qubit in superposition, that are guaranteed to be physically realizable and can be compiled into circuits.

Prior work, however, leaves unknown whether it is possible to program a quantum computer using other general forms of control flow, such as higher-order functions, that underpin expressive classical languages. Answering this question requires identifying the complete set of control flow abstractions that a quantum computer can support -- that is, the necessary and sufficient conditions for control flow in superposition to be correctly realizable -- which has not been done to date.

\subsection{Theoretical Limits of Control Flow in Superposition}

In this work, we provide a complete characterization of the properties of control flow abstractions that are correctly realizable on a quantum computer, showing that many general forms of control flow from classical programming fail to work correctly over data in quantum superposition.

First, we prove that even given a quantum computer endowed with a representation of a program counter in superposition, one cannot correctly realize control flow within a quantum algorithm by directly lifting the classical conditional jump instruction to work on data in superposition. In turn, the typical strategy from classical programming of compiling arbitrary control flow abstractions to conditional jumps can lead a quantum computer to incorrectly execute quantum algorithms.

\paragraph{Landauer Embedding}
To explain why the classical conditional jump does not work on a quantum computer, we first show that the standard technique researchers use to lift a classical computation into a quantum one gives incorrect results for control flow abstractions such as conditional jump.

Fundamentally, the hardware primitives that can be used to realize a quantum computer that supports control flow in superposition are quantum logic gates, which operate on quantum data without measuring it and thereby inappropriately collapsing its superposition. The reason that it is critical for quantum data to not collapse from superposition is to ensure that the data correctly exhibits interference. Without interference, an algorithm such as~\citet{shor1997} would produce a correct output with exponentially smaller probability and thus lose its quantum advantage.

The challenge is that the mathematical semantics of each quantum logic gate is a unitary operator, which is an invertible and therefore injective function.
By contrast, the state transition function of the conventional conditional jump instruction is not injective. When two distinct instructions jump to the same point, then the identity of the machine state before the jump is lost after the jump.

One reason to hope that this problem may be solvable is a technique, developed by researchers in quantum algorithms and complexity theory, to lift a non-injective classical computation to an injective quantum computation. Specifically, one can convert any function $f(x)$ into an injective function $g(x) = (x, f(x))$ that returns a copy of its input. When iterated over an entire computation, this process yields the computation's output alongside a \emph{history} of its intermediate states.

This standard technique is known as \emph{Landauer embedding}~\citep{landauer1961}, so named as it embeds a computation into a larger state space containing the history.
Prior work~\citep{lagana2009} has proposed to use Landauer embedding to implement the non-injective semantics of conditional jump as a unitary operator. In this scheme, the machine stores a history of previous values of the program counter, and appends a value to the history upon executing each jump instruction.

\paragraph{Disruptive Entanglement}
In this work, however, we demonstrate that a quantum computer that uses Landauer embedding to implement conditional jump produces incorrect outputs, which lead a quantum algorithm to not correctly produce interference and to lose its quantum advantage.

The cause is quantum \emph{entanglement}, the phenomenon in which a superposition state cannot be represented as the product of independent components, and in which discarding one component necessarily collapses the superposition of the other. For a program whose control flow depends on data, the machine produces a state in which the history of the program counter is entangled with the output data. Here, being entangled with the history means that the output data fails to correctly produce interference, and the quantum algorithm produces incorrect results.

\paragraph{No-Embedding Theorem}
More generally, we prove that a quantum computer cannot correctly support any form of control flow in superposition, including conditional jump, whose state transition semantics is not injective.
Formally, we generalize the Landauer embedding to the concept of a \emph{quantum embedding} that defines the most general way to lift a classical computation to a quantum one by storing auxiliary information. Then, we prove a \emph{no-embedding theorem} stating that using any such technique to realize control flow causes the computer to produce incorrect outputs.

As a corollary, we prove that a $\lambda$-calculus featuring superpositions of $\lambda$-terms cannot be used to program a quantum computer, as conjectured by~\citet{vantonder2004}. The reason is that $\beta$-reduction is not injective --- both $\lambda x. x$ and $(\lambda x. x)(\lambda x.x)$ reduce to $\lambda x. x$.
This result precludes the Church encoding of quantum information into $\lambda$-terms, in that it prevents the superposition of two bits $\smash{\frac{1}{\sqrt{2}}}(\ket{0} + \ket{1})$ from being represented as a superposition of $\lambda$-terms $\smash{\frac{1}{\sqrt{2}}}(\ket{\lambda x. \lambda y. x} + \ket{\lambda x. \lambda y. y})$.

\paragraph{Implications}
Our theorem provides a unifying explanation of why numerous classical control flow abstractions, ranging from closures to continuations, remain challenging to adapt to quantum programming since their classical basis in conditional jumps or the $\lambda$-calculus cannot be correctly realized in superposition. New abstractions must take their place in quantum programming.

This result raises caution for proposals from the hardware community~\citep{wang2022,meier2024} for a quantum equivalent of the von Neumann architecture in the form of a reprogrammable quantum computer that stores instructions alongside data in superposition. Though experimentalists have attempted to physically realize such a design~\citep{kjaergaard2020}, our theorem implies that it necessarily supports limited forms of control flow compared to classical computers.

To our knowledge, prior designs of quantum von Neumann architectures do not acknowledge or account for the fundamental limitations to control flow that are formalized by our no-embedding theorem. For instance, to implement a conditional jump instruction, the proposal of \citet{lagana2009} uses the approach of histories and hence does not correctly execute quantum algorithms.

\subsection{Specification for Sound Control Flow in Superposition}

The no-embedding theorem implies that one cannot use arbitrary classical abstractions for control flow to correctly program a quantum computer. To codify the properties of forms of control flow that are correct to use on a quantum computer, we next present the necessary and sufficient conditions for control flow in superposition to be correctly realizable as part of a quantum program.

The first correctness condition, \emph{injectivity}, specifies the control flow abstractions that may be correctly realized on a quantum computer. The second, \emph{synchronization}, specifies the valid programs that may be constructed using these abstractions to correctly implement quantum algorithms.

\paragraph{Injectivity}
As stated by the no-embedding theorem, a quantum computer can correctly realize a programming abstraction for control flow in superposition only if its state transition semantics is inherently injective. The semantics cannot use an embedding --- that is to say, it cannot accumulate a history or any other auxiliary information not integral to the control state of the machine.

\paragraph{Synchronization}
Injectivity is strong enough to guarantee that an abstraction is realizable, but not that all programs constructed from it produce correct outputs.
Given a computation in which the program counter exists in superposition, it is possible for the program counter to become entangled with the data registers and remain entangled in the final machine state of the computation. If so, the problem of disruptive entanglement arises once again, meaning the output is incorrect.

To address this issue, the second necessary condition of \emph{synchronization} states that control flow must eventually become separable from, i.e.\ not entangled with, the data. More precisely, a program is synchronized when at the end of execution, the value of the program counter is identical across each classical state within the machine state superposition. Hence, the data registers and program counter are not entangled, making the output data correct for use by the quantum algorithm.

\paragraph{Implications}
Critically, the fact that control flow must be synchronized to be correctly realizable implies that a quantum computer based on the concept of a program counter can support a \texttt{while}-loop whose number of iterations is a data-dependent value in superposition only if that number is bounded by a classical value. Only then can the loop be synchronized and thus correctly realizable.

\subsection{Instruction Set Architecture for Control Flow in Superposition}

To implement the specification above, we present the \emph{quantum control machine}, an instruction set architecture for quantum programming with control flow in superposition. This architecture is physically realizable via quantum logic gates and is the first to provide both a representation of the program counter in superposition and a sound means of manipulating the program counter.

\paragraph{Instruction Set}
Instead of the conventional conditional jump, the machine uses other control flow instructions whose state transition semantics are inherently injective. These instructions, originally introduced by architectures for classical reversible computers~\citep{thomsen2012,axelsen2007}, use reversible arithmetic to manipulate a \emph{branch control} register whose value tracks how much the program counter advances and is added to the program counter after each cycle.

On top of the control flow instructions from prior work, the quantum control machine adds the ability to execute unitary operators that create and manipulate superpositions of data such as the Hadamard gate, and provides a representation of a program counter in superposition.

\paragraph{Case Studies}
The quantum control machine unifies several forms of control flow that can be implemented in existing quantum programming languages.
In a case study that implements core components of quantum algorithms for phase estimation, quantum walk, and physical simulation, we illustrate how a developer can realize the imperative control flow abstractions of branching (\texttt{switch}) and bounded iteration (\texttt{for}) as synchronized programs. The case study demonstrates how existing control flow patterns that appear in quantum algorithms can be represented in a uniform way using the abstraction of a program counter in superposition and its correct manipulation.

\paragraph{Implications}
Rather than as a target for near-term hardware realization, which is likely to be challenging, we view the quantum control machine as a theoretical model for expressing algorithms that reduces reliance on hardware-level logic gates, and as an intermediate compilation target for proposed quantum programming languages with control flow abstractions such as recursion~\citep{ying2014,ying2012} that have to date not been realized directly in terms of circuits.

A problem that remains open is whether an injective analogue of the $\lambda$-calculus could be similarly developed, which would enable functional programming abstractions for control flow in superposition. The challenge is that in general, function application is not injective, as there is no general way to turn the result of a function application back into a pair of the function and its argument. Moreover, any substitution-based model of computation in which an expression reduces to a final value, and that value also reduces to itself, is fundamentally not injective.

\subsection{Contributions}

In this work, we present the following contributions:
\begin{itemize}[leftmargin=5.5mm]
  \item (\Cref{sec:examples}) We identify that the Landauer embedding, a standard approach to lift classical to quantum computation, does not correctly realize a conditional jump instruction in superposition.
  \item (\Cref{sec:main-theorem}) We prove that programming abstractions with non-injective transition semantics, such as the conventional conditional jump or the $\beta$-reduction of $\lambda$-calculus, cannot correctly realize control flow in superposition, thereby proving conjecture of \citet{vantonder2004}.
  \item (\Cref{sec:synchronization}) We define the necessary and sufficient conditions for control flow in superposition to be correctly realizable as part of a quantum program. First, each programming abstraction must have injective state transition semantics. Second, a program must not entangle the states of data and control flow in its final output, a condition we term synchronization.
  \item (\Cref{sec:machine,sec:case-study}) We introduce the quantum control machine, an instruction set architecture that correctly supports imperative abstractions for branching and bounded iteration in superposition, and present a case study that uses the machine to implement a range of quantum algorithms.
\end{itemize}

\paragraph{Summary}
In this work, we reveal limits on our ability to lift foundational programming abstractions such as the conditional jump and the $\lambda$-calculus to work with data in quantum superposition, a stark contrast to the historical development of control flow abstractions in classical programming.
Faced with these limits, we propose sound principles for using control flow in quantum programs. These principles underpin a new instruction set architecture that paves way to more convenient theoretical models for quantum algorithms and more expressive quantum programming languages.

\section{Background on Quantum Computation} \label{sec:background}

This section overviews key concepts in quantum computation relevant to this work. For a comprehensive reference in quantum computation, please see~\citet{nielsen_chuang_2010}.

\paragraph{Superposition}
A \emph{qubit} exists in a \emph{superposition} of the classical states 0 and 1 --- a linear combination $\gamma_0 \ket{0} + \gamma_1 \ket{1}$ where $\gamma_0, \gamma_1 \in \mathbb{C}$ are complex \emph{amplitudes} satisfying $\abs{\gamma_0}^2 + \abs{\gamma_1}^2 = 1$.
Examples of qubits are classical $\ket{0}$ and $\ket{1}$, as well as the states $\smash{\frac{1}{\sqrt{2}}}{(\ket{0} + e^{i\varphi} \ket{1})}$ where $\varphi \in [0, 2\pi)$ is known as a \emph{phase}. More precisely, $\varphi$ is the phase of $\ket{1}$ \emph{relative} to $\ket{0}$. By contrast, two states that differ only by a \emph{global} phase, such as $\ket{0}$ and $e^{i\varphi}\ket{0}$, are physically indistinguishable and considered equivalent.

\paragraph{Quantum State}
More generally, a \emph{quantum state} $\ket{\psi}$ of dimension $2^n$ is a superposition over $n$-bit strings. For example, $\ket{\psi} = \smash{\frac{1}{\sqrt{2}}}{(\ket{00}+\ket{11})}$ is a quantum state over two qubits.

The set of $2^n$-dimensional quantum states constitutes the Hilbert space, i.e.\ the formal vector space, $\mathbb{C}^{2^n}$.
A frequently used basis for this space, known as the \emph{computational basis}, is the subset of states $\{\ket{x} \mid x \text{ is an $n$-bit string}\}$ in which one classical state has the entire amplitude 1.

\paragraph{Tensor Product}
Formally, multiple component states form a composite state by the \emph{tensor product} operator $\otimes$. For example, the state $\ket{01}$ is equal to $\ket{0} \otimes \ket{1}$. As is standard in quantum computation, we interchangeably use the notations $\ket{0}\ket{1}$, $\ket{01}$, and $\ket{0, 1}$ to represent $\ket{0} \otimes \ket{1}$.

\paragraph{Physical Operations}
The \emph{norm} of a quantum state $\ket{\psi} = \sum_i \gamma_i \ket{\psi_i}$ is defined as $\norm{\ket{\psi}} = \sum_i \abs{\gamma_i}^2$. A quantum state is physically realizable only if its norm is 1. An operator, i.e.\ function $O$ over states is \emph{norm-preserving} if $\norm{\ket{\psi}} = \norm{O\ket{\psi}}$ for any $\ket{\psi}$, and \emph{linear} if $O(\gamma_1 \ket{\psi_1} + \gamma_2 \ket{\psi_2}) = \gamma_1 (O \ket{\psi_1}) + \gamma_2 (O \ket{\psi_2})$ for any $\gamma_i$ and $\ket{\psi_i}$.
There are exactly two types of operations over quantum states that are physically realizable on a quantum computer --- \emph{unitary operators} and \emph{measurement}.

\paragraph{Unitary Operator}
A \emph{unitary operator} $U$ is a linear and norm-preserving operator over quantum states.
Any unitary operator $U$ satisfies the property that its inverse $U^{-1}$ is equal to its Hermitian adjoint $U^\dagger$.
Formally, a unitary operator may be constructed as a circuit of \emph{quantum gates}. For example, the quantum gates that operate over a single qubit include:
\begin{itemize}
    \item Bit flip (\texttt{X} or \texttt{NOT}), which maps $\ket{x} \mapsto \ket{1 - x}$ for $x \in \{0, 1\}$;
    \item Phase flip (\texttt{Z}), which maps $\ket{x} \mapsto (-1){}^x\ket{x}$;
    \item Hadamard (\texttt{H}), which maps $\ket{x} \mapsto \frac{1}{\sqrt{2}}{(\ket{0} + (-1){}^x \ket{1})}$.
\end{itemize}
A gate may be \emph{controlled} by one or more qubits, forming a larger unitary operator. For example, the two-qubit \texttt{CNOT} gate maps $\ket{0}\ket{x} \mapsto \ket{0}\ket{x}$ and $\ket{1}\ket{x} \mapsto \ket{1}\texttt{NOT}\ket{x} = \ket{1}\ket{1-x}$.

\paragraph{Measurement}
Measuring a quantum state probabilistically collapses its superposition into a classical outcome.\footnotemark{}
Measuring a qubit $\gamma_0 \ket{0} + \gamma_1 \ket{1}$ yields $0$ with probability $\abs{\gamma_0}^2$ and $1$ with probability $\abs{\gamma_1}^2$.
Selectively measuring a state yields a partial outcome and an unmeasured remainder. For example, measuring the first qubit in the state $\ket{\psi} = \gamma_0 \ket{0} \ket{\psi_0}+ \gamma_1 \ket{1} \ket{\psi_1}$ yields the outcome 0 and remainder $\ket{\psi_0}$ with probability $\abs{\gamma_0}^2$, and outcome 1 and remainder $\ket{\psi_1}$ with probability $\abs{\gamma_1}^2$.%
\footnotetext{For ease of understanding, the definition of measurement given here is for a projective measurement in the computational basis. Nevertheless, all results in this work hold equally on the more general definitions of measurement, which can be realized using only unitary operations and projective measurements in the computational basis.}

\paragraph{Copying and Discarding}
The \emph{no-cloning} theorem~\citep{wootters1982} says that no physical process can transform $\ket{\psi} \mapsto \ket{\psi} \otimes \ket{\psi}$ for arbitrary $\ket{\psi}$. Quantum data can be copied only if its basis is fixed --- that is, classical information in computational basis can be copied, whereas arbitrary superpositions cannot.
The \emph{no-deleting} theorem~\citep{pati2000} states that no unitary operator realizes the converse process $\ket{\psi_1} \otimes \ket{\psi_2} \mapsto \ket{\psi_1}$ to delete an arbitrary $\ket{\psi_2}$. In fact, the principle of \emph{implicit measurement}~\citep[Section 4.4]{nielsen_chuang_2010} dictates that discarding a quantum state, i.e.\ throwing it away permanently, is indistinguishable from measuring it.

\paragraph{Entanglement}
Given a product $\ket{\psi} = \ket{\psi_1} \otimes \ket{\psi_2}$, measuring $\ket{\psi_1}$ leaves behind the remainder $\ket{\psi_2}$, and we call such a state $\ket{\psi}$ \emph{separable}.
The opposite of a separable state is an \emph{entangled} state that cannot be written as a tensor product of two components.
Given an entangled state, measuring one component causes the superposition of the other to also collapse.
For example, the \emph{Bell state}~\citep{bell1964} $\ket{\psi} = \smash{\frac{1}{\sqrt{2}}}(\ket{00} + \ket{11})$ is entangled as it cannot be written as a product of two independent qubits. In this state, measuring either of the qubits causes both qubits to collapse from superposition to equal outcomes: either $\ket{0}$ and $\ket{0}$ or $\ket{1}$ and $\ket{1}$ with probability $\big\lvert\frac{1}{\sqrt{2}}\big\rvert{}^2 = \frac{1}{2}$ each.

\paragraph{Interference}
A superposition state fundamentally differs from the distribution of outcomes to which it collapses --- only the former exhibits quantum \emph{interference}, the phenomenon in which the complex amplitudes of a state combine and cancel, as needed by quantum algorithms.

For example, applying the Hadamard gate to the qubit $\ket{0}$ yields the new state $\smash{\frac{1}{\sqrt{2}}}{(\ket{0} + \ket{1})}$ for that qubit. Next, applying the Hadamard gate to that qubit a second time yields:
\begin{align*}
  & \textstyle\frac{1}{\sqrt{2}}{\big(\frac{1}{\sqrt{2}}{(\ket{0} + \ket{1})} + \frac{1}{\sqrt{2}}{(\ket{0} - \ket{1})}\big)} \\
  ={} & \textstyle\frac{1}{2}(\ket{0} \cancel{\textcolor{mygray}{{+} \ket{1}}} + \ket{0} \cancel{\textcolor{mygray}{{-} \ket{1}}}) = \ket{0}
\end{align*}
which when measured always yields 0. Here, interference is the phenomenon that the branches $\ket{1}$ and $-\ket{1}$ with opposite phase mathematically cancel, leaving only the branch $\ket{0}$.

In a quantum algorithm such as integer factorization~\citep{shor1997}, interference is essential to efficiently pruning down a large search space and achieving computational advantage.

\paragraph{Disruptive Measurement}
Given the qubit $\smash{\frac{1}{\sqrt{2}}}{(\ket{0} + \ket{1})}$ after performing the first Hadamard gate, suppose we were to measure it before performing the second Hadamard gate.

Upon measurement, the qubit collapses from superposition to an equal probabilistic mixture of $\ket{0}$ and $\ket{1}$.
When we then apply the second Hadamard gate on this mixture, we obtain not $\ket{0}$ but rather a mixture of $\smash{\frac{1}{\sqrt{2}}}(\ket{0} + \ket{1})$ and $\smash{\frac{1}{\sqrt{2}}}(\ket{0} - \ket{1})$.
Now, no interference occurs --- if we measure this new mixture, we observe an outcome of 0 or 1 with equal probability rather than 0 with certainty.

\paragraph{Disruptive Entanglement}
A key building block for the results in this paper is the fact that the presence of entanglement between the primary state of a computation and an auxiliary or temporary value can cause quantum interference to not occur and quantum advantage to be lost.

To demonstrate, let us start with the same qubit $\smash{\frac{1}{\sqrt{2}}}{(\ket{0} + \ket{1})}$ from above, and then execute the two-qubit \texttt{CNOT} gate on this qubit alongside a new second qubit initialized to $\ket{0}$:
\begin{align*}
  \textstyle\frac{1}{\sqrt{2}}{(\ket{0} + \ket{1})} \otimes \ket{0} \smash{\xmapsto{\texttt{CNOT}}{}} & \textstyle\frac{1}{\sqrt{2}}(\ket{00} + \ket{11})
\end{align*}

The result is the entangled Bell state.
If we were to discard the second qubit, which is equivalent to measuring it, then the first qubit would also collapse from superposition and fail to exhibit interference, as above.
A key fact is that even if we refuse to discard or measure the second qubit, and simply apply the Hadamard gate again to the first qubit, interference still fails to occur:\footnote{Familiar readers may see this case as the deferred measurement principle~\citep{nielsen_chuang_2010}.}%
\begin{align*}
  & \textstyle\frac{1}{\sqrt{2}}{\big(\frac{1}{\sqrt{2}}{(\ket{0} + \ket{1})} \otimes \ket{0} + \frac{1}{\sqrt{2}}{(\ket{0} - \ket{1})} \otimes \ket{1} \big)} \\
  ={} & \textstyle\frac{1}{2}(\ket{00} + \ket{10} + \ket{01} - \ket{11}) \neq \ket{0} \otimes \ket{\psi} \,\text{for any $\ket{\psi}$}
\end{align*}

Unlike $\ket{1}$ and $-\ket{1}$, the branches $\ket{10}$ and $-\ket{11}$ do not interfere, meaning that measuring the first qubit of the final state yields 0 or 1 with equal probability rather than 0 with certainty.

A quantum computation subject to disruptive entanglement degrades to a classical probabilistic computation, which is commonly understood to result in a loss of quantum advantage~\citep[Section 3.2.5]{nielsen_chuang_2010}.
In an algorithm such as~\citet{shor1997}, disrupting interference via entanglement as above causes a wrong answer to be produced or quantum advantage to be lost.

In the example, a correct way to recover the qubit $\smash{\frac{1}{\sqrt{2}}}{(\ket{0} + \ket{1})}$ with its superposition intact from the entangled state $\smash{\frac{1}{\sqrt{2}}}(\ket{00} + \ket{11})$ is to execute the inverse of the \texttt{CNOT} gate, eliminating the undesired entanglement so as to recover the separable state of $\frac{1}{\sqrt{2}}{(\ket{0} + \ket{1})} \otimes \ket{0}$ once again.

\section{Failure of Conditional Jump in Superposition} \label{sec:examples}

In this section, we illustrate the challenges in programming with control flow that depends on data in quantum superposition, and reveal that standard techniques for lifting classical to quantum computation produce incorrect outputs when applied to abstractions for control flow.

\paragraph{Running Example}
Suppose we must implement a program $P$ that, given two machine integer variables \texttt{x} and \texttt{y}, updates their values according to the following transformation:
\begin{equation} \label{eqn:program}
  \ket{\texttt{x}, \texttt{y}} \overset{P}{\mapsto} \begin{cases} \ket{\texttt{x}, \texttt{y} + 1} & \text{if $\texttt{x} = 0$} \\ \ket{\texttt{x} + 1, \texttt{y}} & \text{if $\texttt{x} \neq 0$} \end{cases}
\end{equation}
\newcommand{\sep}{\ensuremath{\!:\!}}%
where $\ket{\texttt{x} \sep 0, \texttt{y} \sep 3}$ denotes a state in which \texttt{x} is 0 and \texttt{y} is 3.
For example, given input $\ket{\texttt{x} \sep 0, \texttt{y} \sep 3}$, the program should yield output $\ket{\texttt{x} \sep 0, \texttt{y} \sep 4}$, and given $\ket{\texttt{x} \sep 3, \texttt{y} \sep 0}$, it should yield $\ket{\texttt{x} \sep 4, \texttt{y} \sep 0}$. We assume that arithmetic operations do not overflow in this example and address overflow in \Cref{sec:overflow}.

Basic operations interleaving arithmetic and control flow such as this example are essential to algorithms for simulation~\citep{babbush2018} and factoring~\citep{shor1997,proos2003} and more generally illustrate the use of control flow in the algorithms we will present in \Cref{sec:case-study}.

\pagebreak

\subsection{Classical Implementation with Conditional Jumps} \label{sec:classical}
\begin{wrapfigure}[8]{r}{.5\textwidth}
\vspace*{-1.5em}
\begin{lstlisting}[language={[x86masm]Assembler}]
    jnz l1 x  ; if x != 0, goto l1
    add y $1  ; add 1 to y
    jmp l2    ; goto l2
l1: add x $1  ; add 1 to x
    jmp l2    ; goto l2
l2: nop       ; no-op
\end{lstlisting}
\setlength{\abovecaptionskip}{5pt}
\caption{Classical assembly implementing \Cref{eqn:program}.} \label{fig:assembly}
\end{wrapfigure}%
On a classical computer, it would be straightforward to realize \Cref{eqn:program} as a transformation of classical data.
In \Cref{fig:assembly}, we depict a typical implementation of this specification as a classical assembly program that relies on conditional jump instructions.

In this assembly program, the machine state contains three registers: \texttt{x}, \texttt{y}, and the program counter \texttt{pc}.
Given the two example initial states of \texttt{x} and \texttt{y} from above, the machine executes the program in \Cref{fig:assembly} by evolving the state according to the following two execution traces:
\begin{align}
\begin{split} \label{eqn:exec-1}
  \ket{\texttt{x} \sep 0, \texttt{y} \sep 3, \texttt{pc} \sep 1} \xmapsto{\texttt{\textbf{jnz}\ l1\ x}} \ket{\texttt{x} \sep 0, \texttt{y} \sep 3, \texttt{pc} \sep 2} \xmapsto{\texttt{\textbf{add}\ y\ \$1}} \ket{\texttt{x} \sep 0, \texttt{y} \sep 4, \texttt{pc} \sep 3} \xmapsto{\texttt{\textbf{jmp}\ l2}} \ket{\texttt{x} \sep 0, \texttt{y} \sep 4, \texttt{pc} \sep 6}
\end{split}\\[0.2em]
\begin{split} \label{eqn:exec-2}
  \ket{\texttt{x} \sep 3, \texttt{y} \sep 0, \texttt{pc} \sep 1} \xmapsto{\texttt{\textbf{jnz}\ l1\ x}} \ket{\texttt{x} \sep 3, \texttt{y} \sep 0, \texttt{pc} \sep 4} \xmapsto{\texttt{\textbf{add}\ x\ \$1}} \ket{\texttt{x} \sep 4, \texttt{y} \sep 0, \texttt{pc} \sep 5} \xmapsto{\texttt{\textbf{jmp}\ l2}} \ket{\texttt{x} \sep 4, \texttt{y} \sep 0, \texttt{pc} \sep 6}
\end{split}
\end{align}
where the first trace is the $\texttt{x} = 0$ branch of \Cref{eqn:program}, and the second is the $\texttt{x} \neq 0$ branch.

\subsection{Superposition of Program Executions} \label{sec:unitarity-problem}

On a quantum computer, we require a program $P$ that manipulates \texttt{x} and \texttt{y} as data that exist in quantum superposition.
Given a superposition of the two input states of \Cref{eqn:exec-1,eqn:exec-2}, the program must produce the corresponding superposition of their output states:
\begin{align} \label{eqn:superposition}
  \textstyle\frac{1}{\sqrt{2}}(\ket{\texttt{x} \sep 0, \texttt{y} \sep 3} - \ket{\texttt{x} \sep 3, \texttt{y} \sep 0}) \overset{P}{\mapsto}{} \textstyle\frac{1}{\sqrt{2}}(\ket{\texttt{x} \sep 0, \texttt{y} \sep 4} - \ket{\texttt{x} \sep 4, \texttt{y} \sep 0})
\end{align}

As a contrast, it would be incorrect for the machine to simply measure \texttt{x} and branch on the outcome. Doing so collapses superposition,\footnote{As a technical note, measuring \texttt{x} in this example also collapses \texttt{y} because \texttt{x} and \texttt{y} are entangled (\Cref{sec:background}), but that fact is not crucial, as the output would be incorrect even if \texttt{x} alone collapsed.} meaning the program produces the output:
\begin{align} \label{eqn:collapse}
\begin{cases} \ket{\texttt{x} \sep 0, \texttt{y} \sep 4} & \text{with probability } \frac{1}{2} \\ \ket{\texttt{x} \sep 4, \texttt{y} \sep 0} & \text{with probability } \frac{1}{2} \end{cases}
\end{align}

If conditional branching collapsed the state as above, it would prevent a quantum algorithm such as \citet{shor1997,ambainis2004,babbush2018} from leveraging interference (\Cref{sec:background}) in order to obtain computational advantage. Specifically, interference cannot possibly occur after the phase stored by the minus sign in \Cref{eqn:superposition} is lost upon collapse.

\paragraph{Program Superposition}
As an alternative to measuring the data, we consider the possibility of a quantum instruction set architecture analogous to the classical one in \Cref{sec:classical} in which the control flow of the program, as embodied by the program counter, may depend on the value of data.

Such a machine is specified as follows. Its state contains quantum registers \texttt{x}, \texttt{y}, and \texttt{pc}, and its state transition function lifts the machine in \Cref{sec:classical} to superposition, taking:
\begin{align*}
  \textstyle\sum_i \gamma_i \ket{\texttt{x}_i, \texttt{y}_i, \texttt{pc}_i} \mapsto \sum_i \gamma_i \ket{\texttt{x}'_i, \texttt{y}'_i, \texttt{pc}'_i}
\end{align*}
whenever the classical machine of \Cref{sec:classical} would step each constituent $\ket{\texttt{x}_i, \texttt{y}_i, \texttt{pc}_i} \mapsto \ket{\texttt{x}'_i, \texttt{y}'_i, \texttt{pc}'_i}$. Mathematically, this operator is \emph{linear} over the quantum state of the machine.

\paragraph{Running Example}
One may envision that on this machine, we could directly execute the program in \Cref{fig:assembly} to manipulate \texttt{x} and \texttt{y} while preserving their superposition.

Given the superposition input from \Cref{eqn:superposition}, this machine sets the initial \texttt{pc} to 1. It then evolves the state according to a superposition of the traces in \Cref{eqn:exec-1,eqn:exec-2}:
\begin{align} \label{eqn:exec-superposition}
\begin{split}
  & \textstyle\frac{1}{\sqrt{2}}(\ket{\texttt{x} \sep 0, \texttt{y} \sep 3, \texttt{pc} \sep 1} - \ket{\texttt{x} \sep 3, \texttt{y} \sep 0, \texttt{pc} \sep 1}) \\[-0.1em]
  \mymapsto{\texttt{\textbf{jnz}\ l1\ x} \ +\ \texttt{\textbf{jnz}\ l1\ x}}{2.5cm} & \textstyle\frac{1}{\sqrt{2}}(\ket{\texttt{x} \sep 0, \texttt{y} \sep 3, \texttt{pc} \sep 2} - \ket{\texttt{x} \sep 3, \texttt{y} \sep 0, \texttt{pc} \sep 4}) \\[-0.1em]
  \mymapsto{\texttt{\textbf{add}\ y\ \$1} \ +\ \texttt{\textbf{add}\ x\ \$1}}{2.5cm} & \textstyle\frac{1}{\sqrt{2}}(\ket{\texttt{x} \sep 0, \texttt{y} \sep 4, \texttt{pc} \sep 3} - \ket{\texttt{x} \sep 4, \texttt{y} \sep 0, \texttt{pc} \sep 5}) \\[-0.1em]
  \mymapsto{\texttt{\textbf{jmp}\ l2} \ +\ \texttt{\textbf{jmp}\ l2}}{2.5cm} & \textstyle\frac{1}{\sqrt{2}}(\ket{\texttt{x} \sep 0, \texttt{y} \sep 4, \texttt{pc} \sep 6} - \ket{\texttt{x} \sep 4, \texttt{y} \sep 0, \texttt{pc} \sep 6}) \\[-0.1em]
  ={} & \textstyle\frac{1}{\sqrt{2}}(\ket{\texttt{x} \sep 0, \texttt{y} \sep 4} - \ket{\texttt{x} \sep 4, \texttt{y} \sep 0}) \otimes \textcolor{mygray}{\cancel{\ket{\texttt{pc} \sep 6}}}
\end{split}
\end{align}

In this trace, each $\mapsto$ represents the execution of a superposition of instructions. For example, the first instance of $\mapsto$ executes a superposition of \texttt{\textbf{jnz}\ l1\ x} and itself, the second executes a superposition of \texttt{\textbf{add}\ y\ \$1} and \texttt{\textbf{add}\ x\ \$1}, and the third a superposition of \texttt{\textbf{jmp}\ l2} and itself.

At the end, the machine may discard the value of \texttt{pc}, leaving behind only the desired output $\frac{1}{\sqrt{2}}(\ket{\texttt{x} \sep 0, \texttt{y} \sep 4} - \ket{\texttt{x} \sep 4, \texttt{y} \sep 0})$ as specified by the right-hand side of \Cref{eqn:superposition}.

\paragraph{Physical Realizability}
However, the ideal semantics in \Cref{eqn:exec-superposition} is not physically realizable on a quantum computer. Physical principles dictate that a transformation over quantum states must take any physically realizable input state to a physically realizable output state. Formally, it must \emph{preserve norms} of states.
By contrast, given certain physically realizable states over \texttt{x}, \texttt{y}, and \texttt{pc}, the state transition function described above can produce a physically unrealizable state:
\begin{align*}
  \begin{split}
  & \textstyle\frac{1}{\sqrt{2}}(\ket{\texttt{x} \sep 0, \texttt{y} \sep 0, \texttt{pc} \sep 3} - \ket{\texttt{x} \sep 0, \texttt{y} \sep 0, \texttt{pc} \sep 5}) \\[-0.1em]
  \mymapsto{\texttt{\textbf{jmp}\ l2}\ +\ \texttt{\textbf{jmp}\ l2}}{2cm}& \textstyle\frac{1}{\sqrt{2}}(\ket{\texttt{x} \sep 0, \texttt{y} \sep 0, \texttt{pc} \sep 6} - \ket{\texttt{x} \sep 0, \texttt{y} \sep 0, \texttt{pc} \sep 6}) = 0
  \end{split}
\end{align*}

Here, the output has norm 0, a physically impossible outcome.
The reason is that the state transition in \Cref{sec:classical} is not injective --- it maps two distinct inputs to the same output:
\begin{align}
  \ket{\textcolor{mygray}{\texttt{x} \sep 0, \texttt{y} \sep 0}, \texttt{pc} \sep 3} & \xmapsto{\texttt{\textbf{jmp}\ l2}} \ket{\textcolor{mygray}{\texttt{x} \sep 0, \texttt{y} \sep 0}, \texttt{pc} \sep 6} \label{eqn:bad-state-1} \\
  \ket{\textcolor{mygray}{\texttt{x} \sep 0, \texttt{y} \sep 0}, \texttt{pc} \sep 5} & \xmapsto{\texttt{\textbf{jmp}\ l2}} \ket{\textcolor{mygray}{\texttt{x} \sep 0, \texttt{y} \sep 0}, \texttt{pc} \sep 6} \label{eqn:bad-state-2}
\end{align}

By contrast, for an operator over quantum states to be both linear and norm-preserving, it must be unitary (\Cref{sec:background}), meaning it has an inverse and is injective by definition.

\subsection{Landauer Embedding and Disruptive Entanglement} \label{sec:entanglement-problem}

There exists a standard technique to compute a non-injective function inside a quantum computation, known as Landauer embedding~\citep{landauer1961}. We show, however, that using this technique to implement control flow leads a quantum algorithm to produce incorrect outputs.

\begin{definition}[Landauer Embedding]
Given a non-injective function $f$, one may \emph{embed} it into an injective function $F(s) = (s, f(s))$ that also returns a copy of the input. From $F(s)$, one extracts the embedded value of $f(s)$ by discarding $s$. This process can be iterated as necessary.
\end{definition}

\paragraph{History}
One may attempt to apply Landauer embedding to the semantics of conditional jump by maintaining a \emph{history} of program counters in memory, into which each executed step is written.
We denote the current program counter by $\texttt{pc}_0$ and the history after $t$ time steps have elapsed by $\texttt{pc}_{1}, \texttt{pc}_{2}, \ldots, \texttt{pc}_{t}$, such that $\texttt{pc}_1$ is the value from the immediately previous time step, and so on.

\paragraph{Physical Realizability}
This construction is realizable as a unitary operator.
One can see that under Landauer embedding, the analogues of the states from \Cref{eqn:bad-state-1,eqn:bad-state-2} evolve differently:
\begin{align*}
\begin{split}
  & \ket{\textcolor{mygray}{\texttt{x} \sep 0, \texttt{y} \sep 0}, \texttt{pc}_0 \sep 3, \ldots, \texttt{pc}_{t} \sep 1} \xmapsto{\texttt{\textbf{jmp}\ l2}} \ket{\textcolor{mygray}{\texttt{x} \sep 0, \texttt{y} \sep 0}, \texttt{pc}_{0} \sep 6, \texttt{pc}_{1} \sep 3, \ldots, \texttt{pc}_{t+1} \sep 1}
\end{split}\\[0.1em]
\begin{split}
  & \ket{\textcolor{mygray}{\texttt{x} \sep 0, \texttt{y} \sep 0}, \texttt{pc}_0 \sep 5, \ldots, \texttt{pc}_{t} \sep 1} \xmapsto{\texttt{\textbf{jmp}\ l2}} \ket{\textcolor{mygray}{\texttt{x} \sep 0, \texttt{y} \sep 0}, \texttt{pc}_{0} \sep 6, \texttt{pc}_{1} \sep 5, \ldots, \texttt{pc}_{t+1} \sep 1}
\end{split}
\end{align*}

\paragraph{Disruptive Entanglement}
Landauer embedding, however, introduces a new problem that causes the computation to produce incorrect results.
Consider the new execution trace of the program in \Cref{fig:assembly}, which updates the old trace in \Cref{eqn:exec-superposition} to add a history of program counters:
\begin{align}
\begin{split} \label{eqn:bad-exec}
  \textstyle\frac{1}{\sqrt{2}}(\!&\ket{\texttt{x} \sep 0, \texttt{y} \sep 3, \texttt{pc}_0 \sep 1} - \ket{\texttt{x} \sep 3, \texttt{y} \sep 0, \texttt{pc}_0 \sep 1}) \\[-0.1em]
  \mymapsto{\texttt{\textbf{jnz}\ l1\ x} \ +\ \texttt{\textbf{jnz}\ l1\ x}}{2.5cm} \textstyle\frac{1}{\sqrt{2}}(\!&\ket{\texttt{x} \sep 0, \texttt{y} \sep 3, \texttt{pc}_0 \sep 2, \texttt{pc}_{1} \sep 1} \\[-0.45em]
  -\ & \ket{\texttt{x} \sep 3, \texttt{y} \sep 0, \texttt{pc}_0 \sep 4, \texttt{pc}_{1} \sep 1}) \\[-0.1em]
  \mymapsto{\texttt{\textbf{add}\ y\ \$1} \ +\ \texttt{\textbf{add}\ x\ \$1}}{2.5cm} \textstyle\frac{1}{\sqrt{2}}(\!&\ket{\texttt{x} \sep 0, \texttt{y} \sep 4, \texttt{pc}_0 \sep 3, \texttt{pc}_{1} \sep 2, \texttt{pc}_{2} \sep 1} \\[-0.45em]
  -\ & \ket{\texttt{x} \sep 4, \texttt{y} \sep 0, \texttt{pc}_0 \sep 5, \texttt{pc}_{1} \sep 4, \texttt{pc}_{2} \sep 1}) \\[-0.1em]
  \mymapsto{\texttt{\textbf{jmp}\ l2} \ +\ \texttt{\textbf{jmp}\ l2}}{2.5cm} \textstyle\frac{1}{\sqrt{2}}(\!&\ket{\texttt{x} \sep 0, \texttt{y} \sep 4, \texttt{pc}_0 \sep 6, \texttt{pc}_{1} \sep 3, \texttt{pc}_{2} \sep 2, \texttt{pc}_{3} \sep 1} \\[-0.45em]
  -\ & \ket{\texttt{x} \sep 4, \texttt{y} \sep 0, \texttt{pc}_0 \sep 6, \texttt{pc}_{1} \sep 5, \texttt{pc}_{2} \sep 4, \texttt{pc}_{3} \sep 1}) \\[-0.1em]
  \neq \textstyle\frac{1}{\sqrt{2}}(\!& \ket{\texttt{x} \sep 0, \texttt{y} \sep 4} - \ket{\texttt{x} \sep 4, \texttt{y} \sep 0}) \otimes \ket{\psi} \,\text{for any}\,\ket{\psi}
\end{split}
\end{align}

Like in \Cref{eqn:exec-superposition}, each $\mapsto$ represents the execution of a superposition of instructions.
The difference is that the final state is now entangled, meaning that discarding or, equivalently, measuring one component collapses the superposition of the other (\Cref{sec:background}).
Thus, the machine cannot discard the history without destroying the superposition of \texttt{x} and \texttt{y}. Discarding $\texttt{pc}_1$ yields:
\begin{align*}
\begin{cases} \ket{\texttt{x} \sep 0, \texttt{y} \sep 4, \cancel{\textcolor{mygray}{\texttt{pc}_0 \sep 6, \texttt{pc}_2 \sep 2, \texttt{pc}_3 \sep 1}}} & \text{w.p. } \frac{1}{2} \\ \ket{\texttt{x} \sep 4, \texttt{y} \sep 0, \cancel{\textcolor{mygray}{\texttt{pc}_0 \sep 6, \texttt{pc}_2 \sep 4, \texttt{pc}_3 \sep 1}}} & \text{w.p. } \frac{1}{2} \end{cases}
\end{align*}
which is the same incorrect outcome as having measured \texttt{x} to begin with, as in \Cref{eqn:collapse}.

Moreover, as established in \Cref{sec:background}, simply refusing to measure or discard the history is not an admissible workaround --- using part of an entangled state in place of the right side of \Cref{eqn:superposition} still leads a quantum algorithm to produce a wrong answer or lose computational advantage.

\subsection{No Recovery from Disruptive Entanglement} \label{sec:uncomputation}

From this point, one conceivable avenue to recover from entanglement is to use \emph{uncomputation}~\citep{bennett1973}, the standard technique to erase temporary data in quantum computation.

\begin{definition}[Uncomputation]
Suppose we have non-injective $f$ and $g$, and seek an injective $H$ such that $H(s) = (s, g(f(s)))$. Denote by $F$ and $G$ the Landauer embeddings of $f$ and $g$ respectively.

Composing $F$ and $G$ produces $H(s)$ alongside a temporary value $f(s)$, which we seek to erase. Since $F$ is injective, we execute its partial inverse, denoted $F^\dagger$, thereby \emph{uncomputing} $f(s)$:%
\begin{align*}
s \overset{F}{\mapsto} (s, f(s)) \xmapsto{\textrm{id}\,\otimes\,G} (s, \cancel{f(s)}, g(f(s))) \xmapsto{F^\dagger\,\otimes\,\textrm{id}} (s, g(f(s)))
\end{align*}

Iterating this process in a computation enables all values, except for the initial $s$, to be erased from the machine state. Importantly, the initial $s$ necessarily persists in the state.
\end{definition}

\paragraph{Pitfall}
One may hope to use uncomputation to erase all but the initial program counter from the history in~\Cref{eqn:bad-exec}.
The problem is that uncomputing $f(s)$ both requires and leaves behind the value $s$, and when $f$ is the machine state transition function, the value of $s$ must store not only the value of \texttt{pc} but also that of all data registers on which a jump may depend.

One may suggest modifying the history to store copies of all data registers, but even when doing so is possible,\footnote{As a technical note, copying data registers is not possible in general under the no-cloning theorem (\Cref{sec:background}).} it still does not resolve the problem of entanglement. The initial values of data registers \texttt{x} and \texttt{y}, denoted $\texttt{x}_{\textrm{in}}$ and $\texttt{y}_{\textrm{in}}$ respectively, are now stored in the history and persist in the state that is left behind after all possible uncomputation, which remains entangled:
\begin{align*}
  \textstyle\frac{1}{\sqrt{2}}(\!&\ket{\texttt{x} \sep 0, \texttt{y} \sep 4, \texttt{x}_{\textrm{in}} \sep 0, \texttt{y}_{\textrm{in}} \sep 3} - \ket{\texttt{x} \sep 4, \texttt{y} \sep 0, \texttt{x}_{\textrm{in}} \sep 3, \texttt{y}_{\textrm{in}} \sep 0}) \\
  \neq \textstyle\frac{1}{\sqrt{2}}(\!& \ket{\texttt{x} \sep 0, \texttt{y} \sep 4} - \ket{\texttt{x} \sep 4, \texttt{y} \sep 0}) \otimes \ket{\psi} \,\text{for any}\,\ket{\psi}
\end{align*}

\paragraph{No-Embedding Theorem}
One may hope that the presence of entanglement is caused by the specific encoding of the history in Landauer embedding, and can be avoided by storing less information.

In \Cref{sec:main-theorem}, we prove otherwise. By generalizing the arguments above, we show that any attempt to implement an instruction set with a non-injective state transition semantics on a quantum computer necessarily suffers from the problem of disruptive entanglement.
This theorem implies that there is fundamentally no way to lift the conventional conditional jump to superposition and guarantee that the output of the program is correct for use by a quantum algorithm.

\subsection{Quantum Control Machine} \label{sec:machine-example}

Faced with the impossibility of lifting conventional conditional jumps to superposition, we present a new instruction set architecture called the \emph{quantum control machine}. The key idea of the approach is to start over with alternative control flow primitives, originally introduced by classical reversible architectures~\citep{axelsen2007,thomsen2012}, whose transition semantics are inherently injective without the need for Landauer embedding.

\paragraph{Branch Control}
Instead of a history, this machine tracks the difference in \texttt{pc} from the previous instruction in a new \emph{branch control} register~\citep{axelsen2007,thomsen2012} denoted \texttt{br}.
The design of the machine redefines the semantics of each jump instruction to manipulate \texttt{br}, and adds its value to \texttt{pc} after each instruction.
Specifically, a \texttt{\textbf{jmp}} updates \texttt{br}, and the value of \texttt{br} persists until it is changed again. When \texttt{br} is 1, the program executes step by step, and when \texttt{br} is greater than 1, it continually jumps forward.
To resume single-step execution, the program invokes designated instructions, known as \emph{reverse jumps}, that reset \texttt{br} back to 1.

\paragraph{Running Example}
\begin{figure}
\centering
\hspace*{1em}%
\begin{minipage}[t]{.47\textwidth}
\begin{lstlisting}[language={[x86masm]Assembler},morekeywords={rjmp,rjz}]
l0: jnz l2 x  ; if x != 0, goto l2
    add y $1  ; add 1 to y
l1: jmp l3    ; jump to l3
l2: rjmp l0   ; come from l0
    add x $1  ; add 1 to x
l3: rjz l1 x  ; if x = 0, come from l1
    nop       ; no-op
\end{lstlisting}
\end{minipage}%
\begin{minipage}[t]{.5\textwidth}
\[\arraycolsep=1.7pt\def\arraystretch{1.2}%
\begin{array}{c l}
& \ket{\texttt{x} \sep 3, \texttt{y} \sep 0, \texttt{pc} \sep 1, \texttt{br} \sep 1} \\
\xmapsto{\minwidthbox{\texttt{\textbf{jnz}\ l2\ x}}{1.2cm}} & \ket{\texttt{x} \sep 3, \texttt{y} \sep 0, \texttt{pc} \sep 4, \texttt{br} \sep 3} \\
\xmapsto{\minwidthbox{\texttt{\textbf{rjmp}\ l0}}{1.2cm}} & \ket{\texttt{x} \sep 3, \texttt{y} \sep 0, \texttt{pc} \sep 5, \texttt{br} \sep 1} \\
\xmapsto{\minwidthbox{\texttt{\textbf{add}\ x\ \$1}}{1.2cm}} & \ket{\texttt{x} \sep 4, \texttt{y} \sep 0, \texttt{pc} \sep 6, \texttt{br} \sep 1} \\
\xmapsto{\minwidthbox{\texttt{\textbf{rjz}\ l1\ x}}{1.2cm}} & \ket{\texttt{x} \sep 4, \texttt{y} \sep 0, \texttt{pc} \sep 7, \texttt{br} \sep 1}
\end{array}\]
\end{minipage}

\setlength{\abovecaptionskip}{5pt}
\begin{minipage}[t]{.5\textwidth}
\caption{Assembly program for the quantum control machine implementing \Cref{eqn:program} using reverse jumps.} \label{fig:qcm-assembly}
\end{minipage}%
\hspace*{1em}%
\begin{minipage}[t]{.47\textwidth}
\caption{Execution trace of the program in \Cref{fig:qcm-assembly}, given the input state $\ket{\texttt{x} \sep 3, \texttt{y} \sep 0}$ from \Cref{eqn:exec-2}.} \label{fig:qcm-assembly-trace}
\end{minipage}
\end{figure}
In \Cref{fig:qcm-assembly}, we present an implementation of \Cref{eqn:program} on the quantum control machine.
The notable differences from \Cref{fig:assembly} are the new reverse jump instructions \texttt{\textbf{rjmp}} and \texttt{\textbf{rjz}}.
To illustrate how these instructions work, in \Cref{fig:qcm-assembly-trace} we depict the execution trace of the program in \Cref{fig:qcm-assembly} given the input state $\ket{\texttt{x} \sep 3, \texttt{y} \sep 0}$ from \Cref{eqn:exec-2}.

First, \texttt{\textbf{jnz}\ l2\ x} increases \texttt{br} from 1 to 3. The value of \texttt{pc} advances by 3 from 1 to 4, and \texttt{br} remains 3. If nothing else is done, on the next iteration \texttt{pc} would jump again from 4 to 7.

However, the next instruction is a reverse jump, \texttt{\textbf{rjmp}\ l0}. The effect of this instruction is precisely the opposite of the jump from \texttt{l0} --- it decreases \texttt{br} from 3 to 1. The rest of the program then executes step by step, where \texttt{\textbf{rjz}\ l1\ x} has no effect when \texttt{x} is not 0.

\paragraph{Physical Realizability}
This transition function is injective and unitary, even without employing Landauer embedding.
The reason is that jumps from different lines to the same line must leave behind distinct values of \texttt{br}, as can be seen using the examples from \Cref{eqn:bad-state-1,eqn:bad-state-2}:
\[\arraycolsep=1.7pt\def\arraystretch{1.2}%
\begin{array}{l c l}
  \ket{\textcolor{mygray}{\texttt{x} \sep 0, \texttt{y} \sep 0}, \texttt{pc} \sep 3, \texttt{br} \sep 1} & \xmapsto{\minwidthbox{\text{jump by +3}}{1.3cm}} & \ket{\textcolor{mygray}{\texttt{x} \sep 0, \texttt{y} \sep 0}, \texttt{pc} \sep 6, \texttt{br} \sep 3} \\
  \ket{\textcolor{mygray}{\texttt{x} \sep 0, \texttt{y} \sep 0}, \texttt{pc} \sep 5, \texttt{br} \sep 1} & \xmapsto{\minwidthbox{\text{jump by +1}}{1.3cm}} & \ket{\textcolor{mygray}{\texttt{x} \sep 0, \texttt{y} \sep 0}, \texttt{pc} \sep 6, \texttt{br} \sep 1}
\end{array}\]

Furthermore, no matter what instruction comes on the next line, it is impossible for the two states above to both transition to exactly the same state:
\begin{itemize}
  \item an unconditional \texttt{\textbf{jmp}} (or \texttt{\textbf{rjmp}}) adds the same value to \texttt{br} in both states, meaning the values of \texttt{br} cannot both become equal when they are initially different, and
  \item a conditional \texttt{\textbf{jnz}} (or \texttt{\textbf{rjnz}} or jump with any other condition) sees the same values for data registers \texttt{x} and \texttt{y} in both states, meaning it always modifies \texttt{br} the same way in both states.
\end{itemize}

\paragraph{Disruptive Entanglement Avoided}
We now demonstrate how the program in \Cref{fig:qcm-assembly} avoids the entanglement problem and produces output data that is correct for use by a quantum algorithm.
When given the superposition input of \Cref{eqn:superposition}, this program executes as follows:
\begin{align*}
  & \textstyle\frac{1}{\sqrt{2}}(\ket{\texttt{x} \sep 0, \texttt{y} \sep 3, \texttt{pc} \sep 1, \texttt{br} \sep 1} - \ket{\texttt{x} \sep 3, \texttt{y} \sep 0, \texttt{pc} \sep 1, \texttt{br} \sep 1}) \\[-0.1em]
  \mymapsto{\texttt{\textbf{jnz}\ l2\ x}\ +\ \texttt{\textbf{jnz}\ l2\ x}}{2.5cm} & \textstyle\frac{1}{\sqrt{2}}(\ket{\texttt{x} \sep 0, \texttt{y} \sep 3, \texttt{pc} \sep 2, \texttt{br} \sep 1} - \ket{\texttt{x} \sep 3, \texttt{y} \sep 0, \texttt{pc} \sep 4, \texttt{br} \sep 3}) \\[-0.1em]
  \mymapsto{\texttt{\textbf{add}\ y\ \$1}\ +\ \texttt{\textbf{rjmp}\ l0}}{2.5cm} & \textstyle\frac{1}{\sqrt{2}}(\ket{\texttt{x} \sep 0, \texttt{y} \sep 4, \texttt{pc} \sep 3, \texttt{br} \sep 1} - \ket{\texttt{x} \sep 3, \texttt{y} \sep 0, \texttt{pc} \sep 5, \texttt{br} \sep 1}) \\[-0.1em]
  \mymapsto{\texttt{\ \textbf{jmp}\ l3}\ +\ \texttt{\textbf{add}\ x\ \$1}}{2.5cm} & \textstyle\frac{1}{\sqrt{2}}(\ket{\texttt{x} \sep 0, \texttt{y} \sep 4, \texttt{pc} \sep 6, \texttt{br} \sep 3} - \ket{\texttt{x} \sep 4, \texttt{y} \sep 0, \texttt{pc} \sep 6, \texttt{br} \sep 1}) \\[-0.1em]
  \mymapsto{\texttt{\textbf{rjz}\ l1\ x}\ +\ \texttt{\textbf{rjz}\ l1\ x}}{2.5cm} & \textstyle\frac{1}{\sqrt{2}}(\ket{\texttt{x} \sep 0, \texttt{y} \sep 4, \texttt{pc} \sep 7, \texttt{br} \sep 1} - \ket{\texttt{x} \sep 4, \texttt{y} \sep 0, \texttt{pc} \sep 7, \texttt{br} \sep 1}) \\[-0.1em]
  ={} & \textstyle\frac{1}{\sqrt{2}}(\ket{\texttt{x} \sep 0, \texttt{y} \sep 4} - \ket{\texttt{x} \sep 4, \texttt{y} \sep 0}) \otimes \cancel{\textcolor{mygray}{\ket{\texttt{pc} \sep 7, \texttt{br} \sep 1}}}
\end{align*}

Just as in \Cref{eqn:exec-superposition}, each $\mapsto$ in the trace denotes the execution of a superposition of instructions. In particular, the final $\mapsto$ executes a superposition of \texttt{\textbf{rjz}\ l1\ x} and itself. On both branches of the superposition, this instruction causes \texttt{pc} to become 7 and \texttt{br} to become 1, but by different means. On the branch where \texttt{x} is 0, \texttt{br} decreases from 3 to 1, and on the branch where \texttt{x} is 3, \texttt{br} starts as 1 and does not change. Note that the transition function is injective thanks to the presence of distinct values of \texttt{x} and \texttt{y}, as $\ket{\texttt{x} \sep 0, \texttt{y} \sep 4, \texttt{pc} \sep 7, \texttt{br} \sep 1}$ and $\ket{\texttt{x} \sep 4, \texttt{y} \sep 0, \texttt{pc} \sep 7, \texttt{br} \sep 1}$ are distinct states.

At the end of this execution, the final machine state is separable, i.e.\ not entangled, and discarding \texttt{pc} and \texttt{br} leaves behind the state of \texttt{x} and \texttt{y} with superposition intact. Thus, the program produces the output that is specified by \Cref{eqn:superposition} and correct for use by a quantum algorithm.

\paragraph{Synchronization}
The example illustrates how a program avoids the entanglement problem by ensuring that at the end of execution, \texttt{pc} and \texttt{br} are equal again across all branches of the state superposition. To do so, the program must possess an appropriate reverse jump instruction at the target of each forward jump instruction according to the structure of its control flow.

In \Cref{sec:synchronization}, we define the class of \emph{synchronized} programs, which include \Cref{fig:qcm-assembly}, that satisfy this condition and thereby guarantee that their output is correct for use by a quantum algorithm. Moreover, whereas the analysis above depicts detailed formal reasoning to show that a program is synchronized, in \Cref{sec:case-study} we demonstrate how in practice a developer may leverage the structure of control flow within a program to more easily identify the program as synchronized.

\paragraph{Existing Methods}
In principle, \Cref{eqn:program} may be implemented using an existing quantum programming language such as Quipper~\citep{quipper}, QWire~\citep{qwire}, or Q\#~\citep{qsharp}.
These languages expose the abstractions of qubits and bit-controlled logic gates that enable the developer to build a representation of the program as a quantum circuit.

It can be challenging, however, to represent control flow in superposition directly via a quantum circuit of logic gates. When expressed as an explicit circuit, a conditional branch over a variable \texttt{x} corresponds to a sequence of logic gates controlled on individual qubits that realize the comparison with \texttt{x}. We illustrate this complexity in \Cref{sec:qsharp-silq}, in which we present a 15-line Q\# program that implements \Cref{eqn:program} by instantiating the 229-gate circuit in \Cref{fig:circuit}.
\begin{figure}
\includegraphics[width=\textwidth]{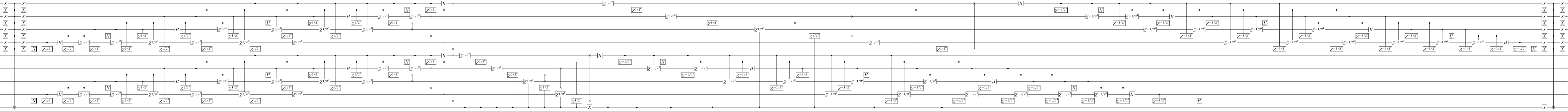}
\caption{Quantum circuit for \Cref{eqn:program}, restricted to 8-bit integers, as implemented by a Q\# program.} \label{fig:circuit}
\end{figure}

To alleviate the complexity of building circuits, researchers have developed higher-level languages such as Silq~\citep{silq}, which provides a quantum \texttt{if}-statement that branches on the value of a qubit in superposition.
In \Cref{sec:qsharp-silq}, we show that the Silq implementation of \Cref{eqn:program} using quantum \texttt{if} also uses the \texttt{forget} statement, an unsafe operation whose correctness must be shown using reasoning outside the immediate automated capabilities of Silq's type system.
The condition that \texttt{forget} must use for the Silq program to be correct is analogous to the condition that the reverse jump instruction must use for the program in \Cref{fig:qcm-assembly} to be synchronized.

Thus, by formally characterizing the properties of sound and realizable control flow on a quantum computer, our work generalizes the reasoning that type systems of existing quantum programming languages perform to enforce injectivity and synchronization.

\paragraph{Summary}
In this section, we presented examples of programs that are important for quantum algorithms and easy to implement on a classical computer, yet are hard to implement on a quantum computer without the abstraction of control flow that can depend on data in superposition.

We illustrated why one cannot program a quantum computer using the conventional conditional jump instruction, and presented an alternative -- reversible jumps and synchronization in the quantum control machine -- that enables correct programming with control flow in superposition.

\section{Theoretical Limits of Control Flow in Superposition} \label{sec:theory}

In this section, we present the \emph{no-embedding theorem}, a no-go theorem for the design of quantum programming languages that states that no programming abstraction with non-injective transition semantics, such as the conventional conditional jump and the $\lambda$-calculus, can be correctly realized in superposition.
More generally, we present the necessary and sufficient conditions for the correct realizability of control flow in superposition. Alongside injectivity, we introduce \emph{synchronization}, the property that the state of data is separable from the state of control flow in the final output of a computation, so that the output data can be correctly used by a quantum algorithm.

\subsection{No-Embedding Theorem and Injectivity} \label{sec:main-theorem}

We begin by defining a \emph{transition system}~\citep{hines2008}, a formalization of any classical model of computation as a set of computation states $S$ and a partial transition function $T : S \rightharpoonup S$.
For example, the Turing machine has $S$ as its state-head-tape configurations and $T$ as its state transition function, while the $\lambda$-calculus has $S$ as the set of $\lambda$-terms and $T$ as $\beta$-reduction.

Though $S$ may be a countably infinite set in principle, any physical computer has a bounded state space in reality. In this section, we assume that all sets are finite and bounded by a function of the length of the longest computation that is physically feasible on the computer.

\paragraph{Superposition of Data}
We lift the definition of transition system to operate over data in quantum superposition.
Formally, we generalize the definitions of \Cref{sec:background} from concrete bit strings to abstract states. Given a finite set of classical states $X$, we denote the corresponding Hilbert space of quantum states as $\mathcal{H}_X$, spanned by the basis $\{\ket{x} \mid x \in X\}$.
Then, a transition system acts on the Hilbert space $\mathcal{H}_S$ via the linear mapping $\mathcal{T}$ that takes $\ket{s} \mapsto \ket{s'}$ whenever $T(s) = s'$ and takes $\ket{s}$ to an unspecified value if $T(s)$ is not defined.

\paragraph{Physical Realizability}
Physically realizing a transition system that operates over quantum information requires constructing it as a unitary operator over an appropriate Hilbert space.
When $T$ is an injective function, $\mathcal{T}$ corresponds to a unitary operator over the Hilbert space $\mathcal{H}_S$.

By contrast, as established in \Cref{sec:examples}, the transition $T$ of the conventional conditional jump is not injective, meaning that it cannot be realized as a unitary operator over $\mathcal{H}_S$.
In the non-injective case, the only alternative possibility is to embed $T$ into a Hilbert space of higher dimension.

\begin{example}[Landauer Embedding]
Let the set $L$ be of \emph{histories} --- lists of elements of $S$.
Then, the larger Hilbert space $\mathcal{H}_S \otimes \mathcal{H}_L$ contains a unitary $U$ that embeds the classical behavior of $T$ on $S$.

Specifically, whenever $T(s) = s'$, we may define $U \ket{s, \ell} = \ket{s', \ell'}$ where $\ell' = \texttt{append}(\ell, s)$.
In the Landauer embedding, computation starts with the empty history $\ell_0 = []$, and discarding, i.e.\ measuring the final $\ell'$ enables a classical output $s' \in S$ to be extracted from the system.
\end{example}

More generally, the use of an auxiliary space makes it possible to embed the classical behavior of a transition system. The Landauer embedding is a specific instance of the following concept:

\begin{definition}[Classical Embedding]
A \emph{classical embedding} for $(S, T)$ is a triple $(\mathcal{H}_L, U, \ket{\eta_0})$ where $\mathcal{H}_L$ is an auxiliary Hilbert space, $\ket{\eta_0}$ is a quantum state with norm 1 that is fixed in $\mathcal{H}_L$, and $U$ is a unitary operator over $\mathcal{H}_S \otimes \mathcal{H}_L$ such that for any $s, s' \in S$ where $T(s) = s'$, we have $U \ket{s, \eta_0} = \ket{s', \eta'}$ for some $\ket{\eta'} \in \mathcal{H}_L$ with norm 1.
\end{definition}

In the above definition, the fact that $\ket{\eta_0}$ must be fixed in $\mathcal{H}_L$ is essential to avoid the possibility of cheating by, for example, setting $\ket{\eta_0} = \ket{s'}$ given advanced knowledge of $s$.

Next, we generalize the above definition to account for not only how a transition system takes classical inputs to classical outputs, but also how in a quantum computation it should accept a superposition of input states and produce a superposition of corresponding output states:

\begin{definition}[Quantum Embedding] \label{def:quantum-embedding}
  A \emph{quantum embedding} for $(S, T)$ is a triple $(\mathcal{H}_L, U, \ket{\eta_0})$ where $\mathcal{H}_L$ is an auxiliary Hilbert space, $\ket{\eta_0}$ is a quantum state with norm 1 that is fixed in $\mathcal{H}_L$, and $U$ is a unitary over $\mathcal{H}_S \otimes \mathcal{H}_L$ such that for any $\ket{\psi}, \ket{\psi'} \in \mathcal{H}_S$ where $\mathcal{T}\ket{\psi} = \ket{\psi'}$, we have $U (\ket{\psi} \otimes \ket{\eta_0}) = \ket{\psi'} \otimes \ket{\eta'}$ for some $\ket{\eta'} \in \mathcal{H}_L$ with norm 1.
\end{definition}

Critically, this definition ensures that the output state $\ket{\psi'} \otimes \ket{\eta'}$ is separable, and subsumes the classical embedding, in which the output state is classical and trivially separable.
As established in \Cref{sec:entanglement-problem}, separability is necessary to guarantee that the output data exhibits correct quantum interference so that the quantum algorithm being implemented produces correct results.

However, while this separability requirement is essential for correctness of algorithms, it in fact precludes the existence of a quantum embedding for any non-injective transition system:
\begin{theorem}[No-Embedding] \label{thm:no-embedding}
  If the state transition function $T$ is not injective, then no quantum embedding exists for the transition system $(S, T)$.
\end{theorem}

\begin{proof}
Assume that a quantum embedding $(\mathcal{H}_L, U, \ket{\eta_0})$ exists for the transition system $(S, T)$, where $T$ is not injective. Let $s_1$ and $s_2$ in $S$ with $s_1 \neq s_2$ be such that $T(s_1) = T(s_2) = s'$. Let $\ket{\psi} = \frac{1}{\sqrt{2}}(\ket{s_1} - \ket{s_2})$. By linearity of $\mathcal{T}$, we have $\mathcal{T}\ket{\psi} = \frac{1}{\sqrt{2}}(\mathcal{T}\ket{s_1} - \mathcal{T}\ket{s_2}) = \frac{1}{\sqrt{2}}(\ket{s'} - \ket{s'}) = 0 \eqcolon \ket{\psi'}$. By \Cref{def:quantum-embedding}, we therefore have $U(\ket{\psi} \otimes \ket{\eta_0}) = \ket{\psi'} \otimes \ket{\eta'} = 0 \otimes \ket{\eta'} = 0$ for some $\ket{\eta'} \in \mathcal{H}_L$. This is a contradiction to the assumption that $U$ is unitary and thus norm-preserving.
\end{proof}

In other words, even though one can embed the purely classical component of a non-injective transition system into a unitary operator, one cannot successfully realize the desired semantics of such a transition system over superpositions of data. There exists no general scheme that correctly lifts a programming abstraction with non-injective transition semantics into superposition.

An immediate corollary of the theorem is that a quantum $\lambda$-calculus featuring superpositions of $\lambda$-terms is not physically realizable, because $\beta$-reduction is not injective --- the distinct terms $\lambda x. x$ and $(\lambda x. x)(\lambda x. x)$ reduce to the same term $\lambda x. x$.
This result formalizes an informal claim of~\citet{vantonder2004}, and we give a detailed instantiation of the theorem to this case in~\Cref{sec:app-lambda}.

\subsection{Synchronization} \label{sec:synchronization}

We next present \emph{synchronization}, the property that the state of data is separable from the state of control flow in the final output of a computation. Together, injectivity and synchronization form a complete specification of the forms of control flow that are correctly realizable in superposition.

We consider transition systems in which $T$ is injective and $S = C \times D$, where $C$ and $D$ denote the control state and data state of the model of computation respectively.
By control state, we refer to the component of the machine state that is not an explicit output of the algorithm being implemented, and will be discarded at the end of computation to leave behind the data.

We define a transition system with $T : C \times D \to C \times D$ as having \emph{control flow dependent on data} when there exist $c, c'_1, c'_2 \in C$ and $d_1, d'_1, d_2, d'_2 \in D$ such that $T(c, d_1) = (c'_1, d'_1)$ and $T(c, d_2) = (c'_2, d'_2)$ and $d_1 \neq d_2$ and $c'_1 \neq c'_2$. In other words, the evolution of the control state may depend on the value of data, and the function $T$ does not factor into two independent transitions $C \to C$ and $C \times D \to D$.

\begin{example}
A classical register machine has a program counter $\texttt{pc}$ and a set $R$ of data registers, in which the semantics of an instruction is a function $\{\texttt{pc}\} \times R \to \{\texttt{pc}\} \times R$. Here, $C$ is the state of \texttt{pc} and $D$ is the state of $R$.
Control flow dependent on data manifests in the conditional jump instruction, whose semantics updates \texttt{pc} based on the value of $R$.
\end{example}

Given a termination time and initial quantum states of control and data, a transition system executes a computation to produce a corresponding final quantum state over control and data:

\begin{definition}[Final State] \label{def:final-machine-state}
  Given an initial control state $\ket{\kappa_0} \in \mathcal{H}_C$, a termination time $t \in \mathbb{N}$, and an input data state $\ket{\delta_0} \in \mathcal{H}_D$, the transition system $(S, T)$ performs the unitary $\mathcal{T}$ (defined in \Cref{sec:main-theorem}) for $t$ iterations on the state $\ket{\kappa_0} \otimes \ket{\delta_0}$ to produce the \emph{final state} $\ket{\psi} \in \mathcal{H}_C \otimes \mathcal{H}_D$.
\end{definition}

When control flow can depend on data, the final $\ket{\psi}$ can be an entangled state in general. However, as established in \Cref{sec:entanglement-problem}, the presence of entanglement would make the output data incorrect for use by a quantum algorithm. To prevent this issue, the property of synchronization specifies that in the final output, the state of data is always separable from the state of control flow:

\begin{definition}[Synchronization] \label{def:synchronization}
  The transition system $(S, T)$ is \emph{synchronized} at initial control state $\ket{\kappa_0}$ and termination time $t$ if there exists some final control state $\ket{\kappa'} \in \mathcal{H}_C$ such that for any input $\ket{\delta_0} \in \mathcal{H}_D$, there exists some output $\ket{\delta'} \in \mathcal{H}_D$ such that the final state of the machine after executing $t$ instructions, as defined in \Cref{def:final-machine-state}, is $\ket{\psi} = \ket{\kappa'} \otimes \ket{\delta'}$.
\end{definition}

To show that a transition system is synchronized, by linearity it suffices to show that the final value of $\ket{\kappa'}$ is fixed across all values of $\ket{\delta_0}$ in computational basis, i.e.\ classical input data alone.

Terminologically, we refer to a program as synchronized when the transition system and initial control state corresponding to the program are synchronized. This term is agnostic to whether the model of computation hardcodes its program in the transition function $T$, as does the standard Turing machine, or whether it accepts its program in $\ket{\kappa_0}$, as does the universal Turing machine.

A synchronized transition system expresses a computation over the data state that is a unitary operator corresponding to a quantum circuit without measurement. The converse also holds --- if a transition system is not synchronized, then it does not express a unitary computation over the data.

\begin{theorem}[Soundness]
If a transition system is injective and synchronized, then given any input data $\ket{\delta_0}$, it produces a final state after $t$ instructions in which the output data $\ket{\delta'}$ is separable from the control state. Furthermore, its mapping from input data to output data is a unitary operator.
\end{theorem}

\begin{proof}
If $(S, T)$ is synchronized, then given any $\ket{\delta_0} = \sum_{i} \gamma_i \ket{\delta_i} \in \mathcal{H}_D$ such that $\mathcal{T}^t(\ket{\kappa_0} \otimes \ket{\delta_i}) = \ket{\kappa'} \otimes \ket{\delta'_i}$, we have $\ket{\psi} = \mathcal{T}^t(\ket{\kappa_0} \otimes \sum_{i} \gamma_i \ket{\delta_i}) = \sum_{i} \gamma_i (\ket{\kappa'} \otimes \ket{\delta'_i}) = \ket{\kappa'} \otimes \ket{\delta'}$ where $\ket{\delta'} = \sum_{i} \gamma_i \ket{\delta'_i}$. The mapping $\ket{\delta_0} = \sum_{i} \gamma_i \ket{\delta_i} \mapsto \ket{\delta'} = \sum_{i} \gamma_i \ket{\delta'_i}$ is linear and norm-preserving as $\mathcal{T}$ is unitary.
\end{proof}

\begin{theorem}[Completeness]
If a transition system is not synchronized, then either there exists some input $\ket{\delta_0}$ for which it produces a final state $\ket{\psi}$ after $t$ instructions in which data and control are entangled, or its mapping from input data to output data is not injective and hence not unitary.
\end{theorem}

\begin{proof}
If $(S, T)$ is not synchronized, then there exist $\ket{\delta_A}, \ket{\delta_B} \in \mathcal{H}_D$ with unit norm such that $\ket{\delta_A} \neq \ket{\delta_B}$, $\mathcal{T}^t(\ket{\kappa_0} \otimes \ket{\delta_A}) = \ket{\kappa_A} \otimes \ket{\delta'_A}$, $\mathcal{T}^t(\ket{\kappa_0} \otimes \ket{\delta_B}) = \ket{\kappa_B} \otimes \ket{\delta'_B}$, and $\ket{\kappa_A} \neq \ket{\kappa_B}$. Then, $\mathcal{T}^t (\ket{\kappa_0} \otimes \frac{1}{\sqrt{2}}(\ket{\delta_A} + \ket{\delta_B})) = \frac{1}{\sqrt{2}}(\ket{\kappa_A} \otimes \ket{\delta'_A} + \ket{\kappa_B} \otimes \ket{\delta'_B})$, which is entangled unless $\ket{\delta'_A} = \ket{\delta'_B}$.
\end{proof}

Together, injectivity and synchronization provide a complete specification for the forms of control flow in superposition that are correctly realizable on a quantum computer --- the programming abstractions must have injective semantics, and the program must be synchronized.

\section{Quantum Control Machine} \label{sec:machine}

In this section, we present the \emph{quantum control machine}, an instruction set architecture for quantum programming with control flow in superposition. This architecture provides for the first time both a program counter in superposition and a sound means of manipulating the program counter.

To satisfy the specification for correctly realizable control flow in superposition, the machine uses variants of conditional jump introduced by classical reversible architectures~\citep{thomsen2012,axelsen2007} that possess inherently injective semantics and can be used to build synchronized programs.
On top of that concept, the quantum control machine adds the ability to execute unitary operators that create and manipulate superpositions of data such as the Hadamard gate, and provides a sound representation of a program counter that exists in superposition.

\subsection{Architectural Overview}

The quantum control machine operates over word-sized quantum registers.
Let $k$, the system word size, be large enough for a word to represent a machine integer or an encoding of an instruction.

\paragraph{Program Encoding}
A \emph{program} is a sequence of instructions of length $\ell$ where $0 < \log_2 \ell < k$. We denote by $\iota_i$ the word-size encoding of the $i$th instruction of the program.

As is the case for a classical Turing machine, a specific instance of the quantum control machine executes a specific fixed program.
The machine obtains access to the instruction sequence through a unitary operator $\texttt{P}$, defined such that given integer $1 \le i \le \ell$, we have $\texttt{P}\ket{i, 0} = \ket{i, \iota_i}$.

In principle, the operator $\texttt{P}$ and its inverse may be physically realized via any technology to make classical data available for read-only use in superposition, such as those in the works of~\citet{babbush2018,berry2019,low2018,giovannetti2008}.

\paragraph{Data Registers}
The machine provides a finite number of quantum registers -- binary-encoded qubit arrays of word size $k$ -- that are indexed and addressable by classical names.
The $n$ data registers are named \texttt{r1}, \texttt{r2}, \ldots, \texttt{r}$n$ and are initialized to 0.

\paragraph{Control Unit}
The machine contains three control registers: the program counter \texttt{pc}, initially 0; the branch control register \texttt{br}, signed and initially 1; and the instruction register \texttt{in}, initially 0.

\subsection{Machine Execution}

Execution occurs in cycles consisting of \emph{fetch}, \emph{execute}, and \emph{retire} stages. Each of the three stages has as its mathematical semantics a unitary operator over the Hilbert space of machine states:

\begin{itemize}
\item The fetch stage loads the next instruction to execute. First, it adds the value of \texttt{br} to \texttt{pc}, using a unitary circuit construction for arithmetic~\citep{draper2000,cheng2002,islam2009,rines2018,cho2020}. It then executes the operator $\texttt{P}$ defined above on $\ket{\texttt{pc}, \texttt{in}}$ to load $\iota_\texttt{pc}$, the instruction at the new $\texttt{pc}$, into register \texttt{in}.

As is the case for a classical architecture, the behavior is implementation-defined when the program counter does not point to a valid instruction, i.e.\ when $\texttt{pc} = 0$ or $\texttt{pc} > \ell$.

\item The execute stage updates the registers \texttt{r1} through \texttt{r}$n$ and \texttt{br} according to the semantics of the instruction in the \texttt{in} register, which is specified in \Cref{sec:isa}.

\item The retire stage executes the inverse of \texttt{P} on $\ket{\texttt{pc}, \texttt{in}}$, so that after a cycle, \texttt{in} has value 0.
\end{itemize}

\subsection{Instruction Set} \label{sec:isa}
\begin{table}
  \caption{Core instruction set for the quantum control machine. The notation $x_i$ denotes the $i$th bit of $x$ and $x_{\setminus i}$ denotes $x$ with $i$th bit set to zero. The notation $p$ denotes a signed immediate offset and the notation $U$ denotes a built-in unitary operator. The expression $[y = 0]$ evaluates to 1 if $y = 0$ or 0 otherwise.}
  \resizebox{0.65\textwidth}{!}{
  \begin{tabular}{ r@{\hspace{0.75em}}l@{\hspace{0.75em}}l }
  \toprule
  & Instruction & Semantics \\
  \midrule
  \emph{No-op} & \texttt{nop} & (identity) \\
  \midrule
  \emph{Unitary} & \texttt{u}\ $U$ \texttt{r}$a$ & ${\ket{x}}_{\texttt{r}a} \mapsto {(U \ket{x})}_{\texttt{r}a} \quad U \in \{\texttt{H}, \texttt{NOT}\}$ \\
  \midrule
  \emph{Swap} & \texttt{swap} \texttt{r}$a$\ \texttt{r}$b$ & $\ket{x}_{\texttt{r}a} \ket{y}_{\texttt{r}b} \mapsto \ket{y}_{\texttt{r}a} \ket{x}_{\texttt{r}b}$ \\
  \emph{Get Bit} & \texttt{get} \texttt{r}$a$\ \texttt{r}$b$\ \texttt{r}$c$\ & $\ket{i}_{\texttt{r}a} \ket{0}_{\texttt{r}b} \ket{x}_{\texttt{r}c}\mapsto \ket{i}_{\texttt{r}a} \ket{x_i}_{\texttt{r}b} \ket{x_{\setminus i}}_{\texttt{r}c}$ \\
  \midrule
  \emph{Add} & \texttt{add}\ \texttt{r}$a$\ \texttt{r}$b$ & $\ket{x}_{\texttt{r}a} \ket{y}_{\texttt{r}b} \mapsto \ket{x + y}_{\texttt{r}a} \ket{y}_{\texttt{r}b}$ \\
  & \texttt{add}\ \texttt{r}$a$\ \texttt{r}$a$ & $\ket{x}_{\texttt{r}a} \mapsto \ket{x \times 2}_{\texttt{r}a}$ \\
  \emph{Multiply} & \texttt{mul}\ \texttt{r}$a$\ \texttt{r}$b$ & $\ket{x}_{\texttt{r}a} \ket{y}_{\texttt{r}b} \mapsto \ket{x \times y}_{\texttt{r}a} \ket{y}_{\texttt{r}b} \quad y \neq 0$ \\
  & \texttt{mul}\ \texttt{r}$a$\ \texttt{r}$a$ & $\ket{x}_{\texttt{r}a} \mapsto \ket{x^2}_{\texttt{r}a}$ \\
  \midrule
  \emph{Jump} & \texttt{jmp}\ $p$ & $\ket{x}_\texttt{br} \mapsto \ket{x + p}_\texttt{br}$ \\
  \emph{Conditional Jump} & \texttt{jz}\ $p$\ \texttt{r}$a$ & $\ket{x}_{\texttt{br}} \ket{y}_{\texttt{r}a} \mapsto \ket{x + p \times [y = 0]}_\texttt{br}\ket{y}_{\texttt{r}a}$ \\
  \emph{Indirect Jump} & \texttt{jmp*}\ \texttt{r}$a$ & $\ket{x}_{\texttt{br}} \ket{y}_{\texttt{r}a} \mapsto \ket{x + y}_\texttt{br}\ket{y}_{\texttt{r}a}$ \\
  \bottomrule
  \end{tabular}%
  }
  \label{tbl:isa}
\end{table}
In \Cref{tbl:isa}, we present the core instruction set and the semantics of each instruction, which updates the data registers $\texttt{r}1, \ldots, \texttt{r}n$ and the control register \texttt{br}. The control registers \texttt{pc} and \texttt{in} are updated only by the fetch and retire stages, and not explicitly by any instruction.

\paragraph{No-op}
In \Cref{tbl:isa}, the first instruction \texttt{nop} is a no-op that simply causes the machine to advance to the next cycle, and its semantics is the identity operator over the machine state.

\paragraph{Unitary Gates}
The next instruction $\texttt{u}$ executes a unitary quantum logic gate, specified by its encoded name $U$ from a fixed set of built-in primitive gates, on the specified register $\texttt{r}a$.

Any unitary operator can be approximated to any degree of precision using single-qubit gates that can be controlled by arbitrarily many qubits~\citep{kitaev1997}. Among the gates that are sufficient for this purpose, the choice of gate set is arbitrary in principle.\footnotemark{}
In this work, we use the single-qubit gate set consisting of $\{\texttt{H}, \texttt{NOT}\}$ --- Hadamard and bit-flip gates on the zeroth qubit of the register, whose controlled variants are sufficient to approximate any unitary operator~\citep{shi2002}.%
\footnotetext{As a technical note, since it is not possible to construct a controlled-$U$ gate given only the ability to apply an unknown gate $U$~\citep{nielsen1997,araujo2014}, the set of gates of the machine must at a minimum be fixed and known.}

We note that in the quantum control machine, because the choice of instruction may be controlled on quantum data by the program counter, multi-qubit gates such as \texttt{CNOT} need not be primitive, and can instead be implemented through control flow instructions.

\paragraph{Data and Arithmetic}
The $\texttt{swap}$ instruction swaps the contents of registers $\texttt{r}a$ and $\texttt{r}b$.
The $\texttt{get}$ instruction extracts the $i$th qubit of $\texttt{r}c$ into $\texttt{r}b$ and leaves 0 in place of that qubit in $\texttt{r}c$, where $i$ is a quantum integer stored in the register $\texttt{r}a$. In principle, the \texttt{get} instruction may be realized in hardware via the quantum analogue of random-access memory, as presented by~\citet{giovannetti2008,paler2020,arunachalam2015,matteo2020}.
For a detailed description of the arithmetic instructions \texttt{add} and \texttt{mul} for addition and multiplication and the behavior of these operations under integer overflow, please see \Cref{sec:overflow}.

\paragraph{Control Flow}
The jump instructions of the machine are from classical reversible architectures~\citep{thomsen2012,axelsen2007}.
The unconditional jump instruction $\texttt{jmp}$ adds the signed immediate value $p$ to the branch control register $\texttt{br}$, producing a relative jump. The magnitude of the jump stored by \texttt{br} persists across cycles until reset by another jump instruction.

The conditional jump instruction $\texttt{jz}$ adds $p$ to \texttt{br} if the value of $\texttt{r}a$ is 0, and has no effect on \texttt{br} otherwise.
The indirect jump instruction $\texttt{jmp*}$ adds the value of $\texttt{r}a$ to \texttt{br}.

\paragraph{Derived Instructions}
Each instruction in the table also has a corresponding reverse instruction prefixed by the character \texttt{r}, whose semantics inverts input and output. For example, the instruction $\texttt{ru}$ executes the operator $U^\dagger$ on $\texttt{r}a$. The exceptions are $\texttt{nop}$ and $\texttt{swap}$, which are self-inverse.

In a quantum computation, reverse instructions are used to uncompute (\Cref{sec:uncomputation}) a register after it is no longer useful, by reversing the sequence of operations that produced its value and restoring it to 0. Reverse instructions must be used for this purpose as no instruction can in general erase a register from an arbitrary value to 0, which is not mathematically a unitary operator.

The core instruction set is readily extended with other derived instructions. For example, we use \texttt{jnz} to denote a jump-if-not-zero instruction, the opposite of \texttt{jz}. In this work, we also permit immediates in place of registers and named labels in place of offsets. Named labels may be translated to a signed offset relative to a \texttt{jmp} instruction that is negated for a \texttt{rjmp} instruction.

\subsection{Termination and Measurement} \label{sec:measurement}

Following \Cref{sec:theory}, we represent the machine state as the product of the control state $C$ containing \texttt{pc}, \texttt{br}, and \texttt{in}, and the data state $D$ containing all of the data registers.
We denote the unitary operator corresponding to one execution cycle by $\texttt{E} : \mathcal{H}_C \otimes \mathcal{H}_D \to \mathcal{H}_C \otimes \mathcal{H}_D$. The machine repeats for a total of $t \in \mathbb{N}$ cycles, so that $\texttt{E}^t$ describes the overall evolution of the machine state.
After $t$ cycles, the machine discards, or equivalently measures, the registers \texttt{pc}, \texttt{br}, and \texttt{in}. If desired, an external process then has the opportunity to measure the states of the data registers.

We follow established convention~\citep{deutsch85,bv1997} in performing all measurement at the end of computation, so that the machine evolution is realizable via unitary logic gates alone.
The principle of deferred measurement states that one final measurement is sufficient to express any quantum computation~\citep{nielsen_chuang_2010}.
Though mid-computation measurement could be added as an extension of the design, it is not strictly necessary to express any computation and does not in general produce output data with superposition intact.

\subsection{Synchronization}

As shown in \Cref{sec:machine-example}, the quantum control machine enables the construction of synchronized programs.
We now instantiate the definition of synchronization (\Cref{sec:synchronization}) for this machine.

\paragraph{Inputs and Outputs}
We identify a program with two sets $Z_\textrm{in}$ and $Z_\textrm{out}$ of data registers that are assumed to be 0 at its start and end respectively.
Let $\mathcal{I}$ be the subspace of $\mathcal{H}_D$ where all registers in $Z_\textrm{in}$ are 0 --- the \emph{input} states of $C$.
Similarly, let $\mathcal{O}$ be where those in $Z_\textrm{out}$ are 0 --- its \emph{output} states.

\paragraph{Final State}
Instantiating \Cref{def:final-machine-state}, we say that given an input state $\ket{\delta_0} \in \mathcal{I}$, and a termination time $t \in \mathbb{N}$, the quantum control machine performs the unitary $\texttt{E}^t$ (defined in \Cref{sec:measurement}) on the state $\ket{\texttt{pc} \sep 0, \texttt{br} \sep 1, \texttt{in} \sep 0} \otimes \ket{\delta_0}$ to produce the final state $\ket{\psi} \in \mathcal{H}_C \otimes \mathcal{O}$.

\paragraph{Synchronization}
Instantiating \Cref{def:synchronization}, we say that the machine is synchronized at time $t$ if there exists $x$ so that for any $\ket{\delta_0} \in \mathcal{I}$, there exists $\ket{\delta'} \in \mathcal{O}$ such that $\ket{\psi} = \ket{\textup{\texttt{pc}} \sep x, \textup{\texttt{br}} \sep 1, \textup{\texttt{in}} \sep 0} \otimes \ket{\delta'}$.

\paragraph{Expressiveness}
In principle, one can express any unitary quantum computation as a synchronized program for the quantum control machine given some appropriate set of primitive unitary gates. Any unitary operator can be approximated to arbitrary precision as a polynomial-length circuit in our gate set of Hadamard and arbitrarily controllable \texttt{NOT} gates~\citep{shi2002}, the latter of which can be implemented by a synchronized program using conditional jumps.

In practice, a developer can verify that a program built from structured branching or iteration constructs is synchronized without use of detailed mathematical reasoning, through two insights:

\begin{itemize}
\item To verify that all conditional branches are synchronized, one needs only to check that the target of each conditional jump instruction in the program is a reverse jump that points back to the original jump and has the same semantic condition as the original jump.
\item To verify that all bounded loops are synchronized, one needs only to check that the execution time of the overall program is independent of each quantum variable in the program, which does not require any specific information about the values of the quantum variables.
\end{itemize}

In the next section, we illustrate how to use the above principles to build synchronized programs for a variety of high-level control flow constructs as found in quantum algorithms.

\section{Case Studies} \label{sec:case-study}

In this section, we illustrate how a developer uses the abstractions for control flow in superposition provided by the quantum control machine to express quantum algorithms. Specifically, we show how a developer can implement imperative abstractions for control flow -- analogues of classical \texttt{for}, \texttt{if}, and \texttt{switch} -- as synchronized programs. The case study demonstrates how a developer can represent control flow patterns from existing quantum algorithms and programming languages in a uniform way by correctly manipulating a program counter in superposition.

We have implemented a simulator for the quantum control machine, which accepts a program, input, and runtime $t$, and outputs the machine state after $t$ steps. Implementations of all case study examples as executable programs are packaged with the simulator in the artifact of this paper.

\subsection{Iteration and Phase Estimation} \label{sec:iteration-case-study}

The quantum control machine enables a program to execute a loop for a quantum number of iterations bounded by a classical value, which is integral to the algorithmic building block of quantum phase estimation~\citep{kitaev1995} as used in algorithms for factoring~\citep{shor1997}, simulation~\citep{abrams1997}, and linear algebra~\citep{harrow2009,abrams1999}.

\paragraph{Exponentiation}
In \Cref{fig:exponent-classical}, we present a classical assembly program for exponentiation. Given \texttt{x} and \texttt{y}, it stores $\texttt{x}^\texttt{y}$ into the register \texttt{res}. To do so, it repeatedly decrements a copy of \texttt{y}, multiplying \texttt{res} by \texttt{x} on each iteration using the quantum analogue of a \texttt{for}-loop (\Cref{ex:iteration}).

For ease of understanding, in this example we store the output in an auxiliary register rather than in place and use repeated multiplication rather than squaring as is typical, which would be more efficient but also more difficult to understand as an example program.
\begin{figure}
\centering\vspace*{-0.5em}
\begin{minipage}[t]{.45\textwidth}
\begin{lstlisting}[language={[x86masm]Assembler},morekeywords={rjmp,rjz,rjnz,swap,rmul,radd}]
    add res $1    ; copy 1 into res
    add r1 y      ; copy y into r1
l1: jz l2 r1      ; if r1 == 0, break
    mul res x     ; multiply res by x
    radd r1 $1    ; decrement r1
    jmp l1        ; goto loop start
l2: nop           ; end of program
\end{lstlisting}
\end{minipage}%
\hspace*{0.8em}%
\begin{minipage}[t]{.49\textwidth}
\begin{lstlisting}[language={[x86masm]Assembler},morekeywords={rjmp,rjz,rjnz,swap,rmul,radd,rjne}]
    add res $1    ; copy 1 into res
    add r1 y      ; copy y into r1
l1: rjne l3 r1 y  ; if r1 != y, come from l3
l2: jz l4 r1      ; if r1 == 0, break
    mul res x     ; multiply res by x
    radd r1 $1    ; decrement r1
l3: jmp l1        ; goto loop start
l4: rjmp l2       ; come from l2
\end{lstlisting}
\end{minipage}

\setlength{\abovecaptionskip}{3pt}
\setlength{\belowcaptionskip}{-0.52em}
\begin{minipage}[t]{.5\textwidth}
\caption{Classical program for exponentiation.} \label{fig:exponent-classical}
\end{minipage}%
\begin{minipage}[t]{.5\textwidth}
\caption{Exponentiation with reverse jumps.} \label{fig:exponent-proposed}
\end{minipage}
\end{figure}

\paragraph{Adapted Program}
Adapting this program to the quantum control machine is done by 1) adapting its control flow to use reversible jumps, and then 2) ensuring that the program is synchronized.

The first step is not the main challenge. By leveraging prior work~\citep{yokoyama2008,axelsen2011} that presents reversible variants of structured \texttt{if} and \texttt{while} constructs, we may straightforwardly insert the appropriate reverse jumps. In \Cref{fig:exponent-proposed}, we present the adapted program in which we insert a corresponding reverse jump as the target of every conditional or backward jump.

The main challenge is the second step of ensuring that the resulting program is synchronized, which in the example means that the final values of \texttt{pc} and \texttt{br} are independent of \texttt{x} and \texttt{y}.
We can see this challenge by executing the program in \Cref{fig:exponent-proposed}.
On the input $\ket{\texttt{x} \sep 2, \texttt{y} \sep 1}$, it produces the final state $\ket{\texttt{x} \sep 2, \texttt{y} \sep 1, \texttt{res} \sep 2, \texttt{pc} \sep 8, \texttt{br} \sep 1}$. Likewise, input $\ket{\texttt{x} \sep 2, \texttt{y} \sep 2}$ results in $\ket{\texttt{x} \sep 2, \texttt{y} \sep 2, \texttt{res} \sep 4, \texttt{pc} \sep 8, \texttt{br} \sep 1}$. At first glance, the program seems synchronized --- \texttt{pc} is always 8 and \texttt{br} is 1.

\paragraph{Problem: Tortoise and Hare}
However, the above values of \texttt{pc} are not for the same time step $t$.
The loop from lines 4 to 7 executes once when \texttt{y} is 1, but twice when \texttt{y} is 2. At $t = 10$, the first input has a \texttt{pc} of 8, having exited the loop, but the second input has a \texttt{pc} of 5, starting the second iteration.

One could continue the slower execution until it also reaches line 8, but by that time, the faster execution will have advanced further again.
In a reversible machine semantics, there can exist no concept of a barrier at which the faster execution stops and waits for the slower one. If execution momentarily halts at an instruction that decrements \texttt{br} to 0, then on the next cycle, \texttt{br} would decrement again, meaning \texttt{pc} starts moving again.
In general, if the tortoise never catches up to the hare, then there is no point in time at which execution may be safely terminated.

\paragraph{Solution: Padding}
\begin{figure}
\centering
\hspace*{-0.75em}%
\begin{minipage}[t]{.53\textwidth}
\begin{lstlisting}[language={[x86masm]Assembler},morekeywords={rjmp,rjz,rjnz,swap,rmul,radd,rjne,rjle}]
    add res $1     ; copy 1 into res
    add r1 max     ; copy max into r1
l1: rjne l3 r1 max ; if r1 != max, come from l3
l2: jz l4 r1       ; if r1 == 0, break
l5: jg l7 r1 y     ; if r1 > y, goto l7
    mul res x      ; multiply res by x
l6: jmp l8         ; break
l7: rjmp l5        ; come from l5
    nop            ; no-op padding
l8: rjle l6 r1 y   ; if r1 <= y, come from l6
    radd r1 $1     ; decrement r1
l3: jmp l1         ; goto start of loop
l4: rjmp l2        ; come from l2
\end{lstlisting}
\end{minipage}%
\begin{minipage}[t]{.47\textwidth}
\begin{lstlisting}[language={[x86masm]Assembler},morekeywords={rjmp,rjz,rjnz,swap,rmul,radd,rjne,u}]
    add r1 i     ; copy i into r1
l1: rjne l3 r1 i ; if r1 != i, come from l3
l2: jz l4 r1     ; if r1 == 0, break
    u H c        ; apply H gate to c
l5: jz l7 c      ; if c == 0, goto l7
    add x $1     ; add 1 to x
l6: jmp l8       ; break
l7: rjmp l5      ; come from l5
    radd x $1    ; subtract 1 from x
l8: rjnz l6 c    ; if c != 0, come from l6
    radd r1 $1   ; subtract 1 from r1
l3: jmp l1       ; goto start of loop
l4: rjmp l2      ; come from l2
\end{lstlisting}
\end{minipage}

\setlength{\abovecaptionskip}{3pt}
\setlength{\belowcaptionskip}{-0.75em}
\begin{minipage}[t]{.5\textwidth}
\caption{Synchronized exponentiation.} \label{fig:exponent-synchronized}
\end{minipage}%
\begin{minipage}[t]{.48\textwidth}
\caption{Program implementing a Hadamard walk.} \label{fig:hadamard-walk}
\end{minipage}
\end{figure}
In \Cref{fig:exponent-synchronized}, we present a synchronized program that avoids the problem. This program executes the loop a fixed, rather than data-dependent, number of times.
Its loop body multiplies \texttt{res} by \texttt{x} for only \texttt{y} iterations, and afterward, it executes padding \texttt{nop}s with no effect.

The new argument \texttt{max}, which we require to be classical, upper-bounds the possible values of \texttt{y}.
After \texttt{max} iterations of the loop, each branch stores the correct \texttt{res}, and the values of \texttt{pc} and \texttt{br} are equal across all branches, so the program is synchronized.
Here, padding is not the only possible approach, but it is simple to use and verify as it guarantees that the program is synchronized.

This example demonstrates how the fundamental property of synchronization restricts the space of valid programs on any quantum computer supporting control flow in superposition. In particular, loops without classical upper bounds cannot be synchronized, which is consistent with prior impossibility results in quantum Turing machines --- for more details, see \Cref{sec:related-work}.

\subsection{Branch Interference and Quantum Walk} \label{sec:control-interference-case-study}

\begin{wrapfigure}[20]{r}{.5\textwidth}
\vspace*{-2em}%
\begin{minipage}{.5\textwidth}%
\begin{align*}%
\arraycolsep=0.125em\def\arraystretch{1.2}%
\begin{array}{ r r l }%
& &\!\ket{\texttt{x} \sep 3, \texttt{c} \sep 0, \texttt{pc} \sep 5} \\
  \mapsto{} & \textstyle\frac{1}{\sqrt{2}}(\!&\!\ket{\texttt{x} \sep 3, \texttt{c} \sep 0, \texttt{pc} \sep 9} + \ket{\texttt{x} \sep 3, \texttt{c} \sep 1, \texttt{pc} \sep 6}) \\
  \mapsto{} & \textstyle\frac{1}{\sqrt{2}}(\!&\!\ket{\texttt{x} \sep 2, \texttt{c} \sep 0, \texttt{pc} \sep 5} + \ket{\texttt{x} \sep 4, \texttt{c} \sep 1, \texttt{pc} \sep 5}) \\
  \mapsto{} & \textstyle\frac{1}{2}(\!&\!\ket{\texttt{x} \sep 2, \texttt{c} \sep 0, \texttt{pc} \sep 9} + \ket{\texttt{x} \sep 2, \texttt{c} \sep 1, \texttt{pc} \sep 6} \\[-0.2em]
  & +{}\vphantom{\textstyle\frac{1}{\sqrt{2}}} &\!\ket{\texttt{x} \sep 4, \texttt{c} \sep 0, \texttt{pc} \sep 9} - \ket{\texttt{x} \sep 4, \texttt{c} \sep 1, \texttt{pc} \sep 6}) \\
  \mapsto{} & \textstyle\frac{1}{2}(\!&\!\ket{\texttt{x} \sep 1, \texttt{c} \sep 0, \texttt{pc} \sep 5} + \ket{\texttt{x} \sep 3, \texttt{c} \sep 1, \texttt{pc} \sep 5} \\[-0.2em]
  & +{}\vphantom{\textstyle\frac{1}{\sqrt{2}}} &\!\ket{\texttt{x} \sep 3, \texttt{c} \sep 0, \texttt{pc} \sep 5} - \ket{\texttt{x} \sep 5, \texttt{c} \sep 1, \texttt{pc} \sep 5}) \\
  \mapsto{} & \textstyle\frac{1}{2\sqrt{2}}(\!&\!\ket{\texttt{x} \sep 1, \texttt{c} \sep 0, \texttt{pc} \sep 9} + \ket{\texttt{x} \sep 1, \texttt{c} \sep 1, \texttt{pc} \sep 6} \\[-0.3em]
  & +{} &\!\ket{\texttt{x} \sep 3, \texttt{c} \sep 0, \texttt{pc} \sep 9} \cancel{\textcolor{mygray}{{-} \ket{\texttt{x} \sep 3, \texttt{c} \sep 1, \texttt{pc} \sep 6}}} \\[-0.1em]
  & +{} &\!\ket{\texttt{x} \sep 3, \texttt{c} \sep 0, \texttt{pc} \sep 9} \cancel{\textcolor{mygray}{{+} \ket{\texttt{x} \sep 3, \texttt{c} \sep 1, \texttt{pc} \sep 6}}} \\[-0.1em]
  & -{}\vphantom{\textstyle\frac{1}{\sqrt{2}}} &\!\ket{\texttt{x} \sep 5, \texttt{c} \sep 0, \texttt{pc} \sep 9} + \ket{\texttt{x} \sep 5, \texttt{c} \sep 1, \texttt{pc} \sep 6}) \\
  \mapsto{} & \textstyle\frac{1}{2\sqrt{2}}(\!&\!\ket{\texttt{x} \sep 0, \texttt{c} \sep 0, \texttt{pc} \sep 14} + \ket{\texttt{x} \sep 2, \texttt{c} \sep 1, \texttt{pc} \sep 14} \\[-0.3em]
  & +{} &\!\ket{\texttt{x} \sep 2, \texttt{c} \sep 0, \texttt{pc} \sep 14} + \ket{\texttt{x} \sep 2, \texttt{c} \sep 0, \texttt{pc} \sep 14} \\[-0.1em]
  & -{} &\!\ket{\texttt{x} \sep 4, \texttt{c} \sep 0, \texttt{pc} \sep 14} + \ket{\texttt{x} \sep 6, \texttt{c} \sep 1, \texttt{pc} \sep 14})
\end{array}
\end{align*}
\end{minipage}
\setlength{\abovecaptionskip}{5pt}
\caption{Partial execution trace of \Cref{fig:hadamard-walk}, showing only steps where $\texttt{br} = 1$. In each state, \texttt{i} and \texttt{r1} are uniform across the superposition, and they are not shown.} \label{fig:hadamard-walk-trace}
\end{wrapfigure}
The quantum control machine distinguishes itself from any classical computer by its ability to exhibit quantum interference across control flow paths of the computation. Such interference is essential to the advantage of quantum walk algorithms such as~\citet{ambainis2004,ambainis2010,childs2007,aharonov2001,shenvi2003}.

In \Cref{fig:hadamard-walk}, we present a program that implements a Hadamard walk~\citep{ambainis2001}, adapted from~\citet{ying2014}.
This program loops over \texttt{i} iterations, where the algorithm specifies \texttt{i} to be classical. Each round, the program executes an \texttt{H} gate over a qubit \texttt{c} and then adds or subtracts 1 from \texttt{x} based on \texttt{c}, using the quantum analogue of an \texttt{if}-statement (\Cref{ex:branching}).

\paragraph{Branch Interference}
The way in which this program demonstrates quantum interference is that measuring the final \texttt{x} value it produces yields a substantially different distribution from a classical random walk that on every round moves \texttt{x} in a uniformly random direction.

In particular, letting the initial \texttt{x} and \texttt{i} be 3, the program executes as in \Cref{fig:hadamard-walk-trace}. At the end of the program, measuring the value of \texttt{x} yields outcome 4 with probability only $1/6 \approx 17\%$, as compared to $3/8 \approx 38\%$ for the classical random walk. The reason is that quantum interference cancels two execution paths, corresponding to outcome 4, that have opposite phase.

The correctness of the program relies on the use of injective abstractions for control flow as opposed to writing down a history of the program counter as in \Cref{sec:entanglement-problem}.
For contrast, we depict in \Cref{sec:incorrect-example} the incorrect execution that would result from writing down such a history. The difference is that the analogues of the identical states that cancel in \Cref{fig:hadamard-walk-trace} are instead distinct and do not interfere to cancel. In the final measurement outcome of the incorrect execution, the outcome of 4 occurs with probability $3/8$, the same result as on a classical computer.

\subsection{Indexed Branching and Quantum Simulation}

The quantum control machine enables branching operations to be indexed by data in superposition, in a way analogous to classical array indexing or \texttt{switch}-statements (\Cref{ex:branching}).

In \Cref{fig:majorana}, we present a program that, given two quantum registers \texttt{i} and \texttt{x}, applies a Majorana fermion operator $\ket{\texttt{i}}\ket{\texttt{x}} \mapsto \ket{\texttt{i}} \texttt{Y}_\texttt{i} \cdot \texttt{Z}_{\texttt{i} - 1} \cdots \texttt{Z}_0 \ket{\texttt{x}}$ to them, which is useful to algorithms for simulation of fermionic systems~\citep{babbush2018}. In this operator, \texttt{Y} is a single-qubit Pauli-$Y$ gate as defined in~\citet{nielsen_chuang_2010}, and we assume that the \texttt{Y} and \texttt{Z} (\Cref{sec:background}) gates are supported as primitive unitary gates.

\begin{wrapfigure}[12]{r}{.5\textwidth}
\vspace*{-2.5em}%
\begin{lstlisting}[language={[x86masm]Assembler},morekeywords={rjmp,rjz,rjnz,swap,rmul,radd,u,get,rget,rjne}]
    get i r1 x    ; put bit i of x into r1
    u Y r1        ; apply Y gate to r1
    rget i r1 x   ; put r1 into bit i of x
    add r1 i      ; copy i into r1
l1: rjne l3 r1 i  ; if r1 != i, come from l3
l2: jz l4 r1      ; if r1 == 0, break
    radd r1 $1    ; subtract 1 from r1
    get r1 r2 x   ; put bit r1 of x into r2
    u Z r2        ; apply Z gate to r2
    rget r1 r2 x  ; put r2 into bit r1 of x
l3: jmp l1        ; goto start of loop
l4: rjmp l2       ; come from l2
\end{lstlisting}
\setlength{\abovecaptionskip}{3pt}
\caption{Program for a Majorana fermion operator.} \label{fig:majorana}
\end{wrapfigure}

The program operates as follows.
First, lines 1 to 3 apply the \texttt{Y} gate on the \texttt{i}th qubit of the \texttt{x} register. The following loop then performs the \texttt{Z} gate on each of the qubits $\texttt{i}-1$ through 0 of the \texttt{x} register. For clarity, the loop has not yet been subject to padding as in \Cref{sec:iteration-case-study}, which must still be done if \texttt{i} is in superposition.

Though this program has not been optimized for practical concerns such as qubit and gate usage, it exemplifies a new programming model in which one can work with quantum data via abstractions similar to classical arrays.

\section{Implications and Directions Forward} \label{sec:implications}

In this section, we discuss implications of this work to research in quantum programming languages, computer architecture, and theory of computation.
For a detailed discussion of the practical costs to realizing control flow in superposition in terms of hardware support and program verification, see \Cref{sec:costs}.
For other related work studying designs for quantum $\lambda$-calculi, reversible computation, oblivious computation, and unbounded-time quantum computation, see \Cref{sec:related-work}.

\subsection{Quantum Programming Languages}

Abstractions for control flow in superposition such as quantum \texttt{if}-statements and \texttt{for}-loops have become a value proposition in emerging quantum programming languages~\citep{silq,tower,voichick2023,pal2022}. The no-embedding theorem provides a unifying explanation for why these abstractions, and others such as recursion and continuations, cannot be adapted to superposition by directly lifting the classical conditional jump.

As an example, language designers have repeatedly and independently confronted the fact that a quantum \texttt{if}-statement is not realizable in general if its branches can be arbitrary statements. Proposed solutions have included the dynamic enforcement of an orthogonality judgment~\citep{qml} and the static enforcement of conditions such as preventing the condition of the \texttt{if} from being modified under its branches~\citep{silq,tower}.
This work presents a correctness condition for control flow in superposition that unifies and generalizes prior approaches, which is that the semantics of each abstraction must be injective and the program must be synchronized. In principle, this condition can be enforced at the level of the \texttt{if}-statement or at the more primitive level of assembly, as in the quantum control machine.

An advantage of sound primitives at the assembly level is that they in turn empower generalized reasoning about the space of realizable abstractions and can act as an intermediate compilation target for emerging abstractions. For instance, though researchers have proposed quantum analogues of recursion and closures~\citep{diazcaro2019,ying2012,ying2014}, to date we are not aware of their realization via a compiler or equivalent. We hope that the quantum control machine may act as a compilation target for such proposals and create new opportunities in language design.

\subsection{Quantum Computer Architecture}

Researchers have proposed stored-program quantum architectures, commonly referred to as quantum von Neumann architectures, with suggested benefits for the efficiency~\citep{kjaergaard2020}, realizability~\citep{meier2024}, and security~\citep{wang2022} of the resulting system.

For example,~\citet{kjaergaard2020} experimentally realize a device that uses one set of qubits to parameterize a rotation gate over another set of qubits. That work describes the ``use of quantum instructions to implement a quantum program'' and ``instructions derived from the present quantum state of the processor'' as potentially advantageous in quantum algorithms for semi-definite programming, simulation, and principal component analysis~\citep{kjaergaard2020}.

However, to our best knowledge, these prior designs have not acknowledged the challenges that will be ultimately encountered when making instructions such as conditional jump operate in superposition. For instance, the no-embedding theorem implies that the machine proposed by \citet{lagana2009}, which attempts to provide conditional jump via a history of program counters, does not correctly execute quantum algorithms.
While the designs of~\citet{wang2022,meier2024} lack an operational specification for the control unit as an instruction set, these designs would face the same fundamental limitations on expressible control flow when fully formalized.

\subsection{Theory of Quantum Computation}

A common, and true, maxim in quantum computation is that any classical computation is also realizable on a quantum computer~\citep[Section 1.4.1]{nielsen_chuang_2010}.
By contrast, in this work, we show that a stronger assumption -- any classical programming abstraction is also correctly realizable on a quantum computer -- is false, as seen in the conventional conditional jump.

An implication is that designers of algorithms would benefit from explicitly specifying control flow as part of the state of a quantum computation rather than leaving it as an implementation detail. The realization of a control flow abstraction requires careful reasoning from the language and potentially the programmer to produce a correct output and preserve quantum advantage.

The scope of this implication is over algorithms that transform quantum data and then leverage interference on the output data, which include~\citet{shor1997} and the other examples in \Cref{sec:case-study}.
We note that it may be possible in limited cases for the design of algorithms to preemptively avoid this concern --- for example, the classical oracle component of~\citet{grover1996} produces an output that is promptly uncomputed, and the algorithm instead leverages interference on the input data.

\section{Conclusion}

Researchers have long studied designs for quantum computers to learn how to realize the design in hardware or analyze its theoretical power. This work advocates for a new dimension of study --- how to correctly and intuitively program the computer to implement quantum algorithms.

Studying a quantum computer through the lens of a programmer reveals the danger that trying to implement a quantum algorithm using classical control flow abstractions such as conditional jump can cause the algorithm to produce incorrect results. Put plainly, programming a quantum computer in the same way as a classical one can in fact turn the quantum computer into a classical computer. If so, the computer's quantum advantage and the return on its investment are lost.

Despite these challenges, we believe that control flow in superposition will remain an indispensible abstraction for expressing quantum algorithms. This work makes it possible for the first time to correctly program a quantum computer using the abstraction of a program counter, bringing the vision of making quantum programs as easy to write as classical programs closer to reach.

\section*{Data Availability Statement}
The software that supports~\Cref{sec:case-study} is available on Zenodo~\citep{artifact}.

\section*{Acknowledgements}
We thank Ellie Cheng, Tian Jin, Jesse Michel, Patrick Rall, and Logan Weber for helpful feedback on this work, and also Scott Aaronson, Soonwon Choi, Isaac Chuang, Aram Harrow, Stacey Jeffrey, Bobak Toussi Kiani, and Yuval Sanders for providing references to related work. This work was supported in part by the National Science Foundation (CCF-1751011) and the Sloan Foundation.

\bibliography{biblio.bib}
\vfill
\clearpage

\appendix

\section{Semantics of Arithmetic and Overflow} \label{sec:overflow}

In this section, we discuss the semantics of the example program of \Cref{sec:examples} in the event of integer overflow. We also present in detail the semantics of the arithmetic operations defined in \Cref{sec:machine}.

\subsection{Example Program of Section~\ref{sec:examples}}

When machine integers are used in practice to implement the program $P$ defined in \Cref{eqn:program}, certain invalid but physically possible inputs may cause the addition operation to overflow. Even when this problem arises, it is still possible to use a unitary operator to implement the restricted definition of \Cref{eqn:program} on valid inputs that do not overflow. When given invalid inputs, this unitary operator produces a state that is physically realizable but arbitrary in principle.

For example, suppose that the machine word size $k=1$, that is, all integers are 1-bit. Then, the input $\texttt{x} = 0, \texttt{y} = 1$ is an invalid input to $P$ as it causes the addition $y + 1$ to overflow.
In this case, it is still possible to realize the valid cases of $P$ using a unitary operator, as follows.

Given the invalid input $\ket{\texttt{x} \sep 0, \texttt{y} \sep 1}$, the output is arbitrary in principle and could be defined to be $\ket{\texttt{x} \sep 0, \texttt{y} \sep 0}$. We may generalize and define the entire semantics of $P$ as the following:
\begin{align*}
\ket{\texttt{x} \sep 0, \texttt{y} \sep 0} \overset{P}{\mapsto} \ket{\texttt{x} \sep 0, \texttt{y} \sep 1} \qquad &
\ket{\texttt{x} \sep 0, \texttt{y} \sep 1} \overset{P}{\mapsto} \ket{\texttt{x} \sep 0, \texttt{y} \sep 0} \\
\ket{\texttt{x} \sep 1, \texttt{y} \sep 0} \overset{P}{\mapsto} \ket{\texttt{x} \sep 1, \texttt{y} \sep 0} \qquad &
\ket{\texttt{x} \sep 1, \texttt{y} \sep 1} \overset{P}{\mapsto} \ket{\texttt{x} \sep 1, \texttt{y} \sep 1}
\end{align*}

This function is injective, meaning that the semantics of $P$ as specified above is realizable as a unitary operator. In the 1-qubit case, only the first case of this truth table defines the valid semantics of $P$ where the corresponding addition operation does not overflow. The second, third, and fourth cases are arbitrary, not arising during the execution of a correctly specified program on valid inputs, and serving only to fully define the operator as unitary.

\subsection{Arithmetic Operations of Section~\ref{sec:machine}}

In \Cref{tbl:isa}, we present the arithmetic instructions of the quantum control machine.
Given two distinct registers $\texttt{r}a$ and $\texttt{r}b$, the $\texttt{add}$ instruction adds $\texttt{r}b$ into $\texttt{r}a$, and $\texttt{mul}$ multiplies $\texttt{r}b$ into $\texttt{r}a$. All arithmetic is unsigned, modulo $2^k$, and formally defined only when overflow does not occur. Multiplication is only defined when $\texttt{r}b$ is nonzero, so that it may be realized via reversible logic for arithmetic.
Given two arguments that are the same register, the $\texttt{add}$ or $\texttt{mul}$ instruction is specialized to double or square the value of that register respectively when overflow does not occur.

The definitions of the reverses of the arithmetic instructions are analogous.
Given two distinct registers $\texttt{r}a$ and $\texttt{r}b$, the instruction $\texttt{radd}\ \texttt{r}a\ \texttt{r}b$ subtracts $\texttt{r}b$ from $\texttt{r}a$ if the latter is greater than or equal to the former, and $\texttt{rmul}\ \texttt{r}a\ \texttt{r}b$ divides $\texttt{r}b$ from $\texttt{r}a$ if the latter is a multiple of the former. Given the register $\texttt{r}a$, the instruction $\texttt{radd}\ \texttt{r}a\ \texttt{r}a$ divides $\texttt{r}a$ by two if it is even, and $\texttt{rmul}\ \texttt{r}a\ \texttt{r}a$ takes the square root of $\texttt{r}a$ if it is a square number. On other inputs, which can never result from \texttt{add} or \texttt{mul}, the semantics of \texttt{radd} and \texttt{rmul} respectively are implementation-defined.

Though formally defined above only in the absence of overflow, addition and multiplication may be implementation-defined to account for the possibility of overflow by the same principle as before, which is that the output program state is any that suffices to make the operator unitary overall.
To illustrate for the addition operation, one possible approach would be to add an integer register $z$ whose value becomes nonzero only in the event of overflow, effectively performing standard reversible arithmetic at a higher bit-width over the combined register $\ket{z,x}$:
\begin{align*}
\ket{x}\ket{y}\ket{z} \mapsto
\ket{(x+y)\;\textrm{mod}\; 2^k}\ket{y}\ket{(z + \lfloor(x + y)/ 2^k\rfloor) \;\textrm{mod}\; 2^k}
\end{align*}
where $k$ is the system word size. This operator is unitary, and during the execution of a program without overflow, the register $z$ remains zero always, meaning that in practice it could be reused for other computation. A similar approach could be taken for multiplication, and this approach is one of many possible implementations that could be ultimately used by the machine implementer.

\section{Example Program Implementation in Q\# and Silq} \label{sec:qsharp-silq}

In \Cref{fig:qsharp}, we present an implementation of the example of \Cref{eqn:program} in Q\#~\citep{qsharp}. The program uses the quantum circuit abstractions of qubits and bit-controlled logic gates, including the manipulation of a temporary qubit \texttt{z} that controls operations on \texttt{y} and \texttt{x}.
\begin{figure}
\begin{lstlisting}[morekeywords={operation,Qubit,Unit,let,use,within,apply,Controlled,Adjoint}]
operation f(x: Qubit[], y: Qubit[]) : Unit {
  use z = Qubit() {
    within {
      within {
        ApplyToEachA(X, x);
      } apply {
        Controlled X(x, z);
      }
    } apply {
      Controlled IncrementByInteger([z], (1, LittleEndian(y)));
      within {
        X(z);
      } apply {
        Controlled IncrementByInteger([z], (1, LittleEndian(x)));
} } } }
\end{lstlisting}
\caption{Implementation of \Cref{eqn:program} in Q\#.} \label{fig:qsharp}
\end{figure}

In \Cref{fig:program-simpler-silq}, we present an implementation of \Cref{eqn:program} in Silq~\citep{silq}. The condition relating \texttt{x} and \texttt{x'} that must be used inside the \texttt{forget} statement for \Cref{fig:program-simpler-silq} to be correct is analogous to the condition on \texttt{x} of the \texttt{\textbf{rjz}} instruction on line 6 in \Cref{fig:qcm-assembly}.

\begin{figure}
\begin{lstlisting}[morekeywords={def,uint,mfree,if,else,forget,then,return,int}]
def f(x: uint[8], y: uint[8]) mfree: uint[8]^2 {
  if x=0 {
    y += 1;
    x' := 0:int[8];
  } else {
    x' := x+1;
  }
  forget(x=if x'=0 then 0:int[8] else x'-1);
  return (x', y);
}
\end{lstlisting}
\caption{Implementation of \Cref{eqn:program} in Silq.} \label{fig:program-simpler-silq}
\end{figure}

\section{Quantum Lambda-Calculus with Quantum Control Is Not Physically Realizable} \label{sec:app-lambda}

In this section, we show that the no-embedding theorem (\Cref{thm:no-embedding}) implies that the quantum $\lambda$-calculus with quantum control is not physically realizable.
This calculus, named $\lambda_i$ by~\citet{vantonder2004}, is a quantum extension of the classical $\lambda$-calculus. The hallmark of its design is that it permits program terms to exist in superposition, unlike other designs for quantum $\lambda$-calculi~\citep{selinger2004} that require program terms and evaluation to be classical.

\paragraph{Syntax}
The calculus extends classical $\lambda$-calculus with unitary operators over $\lambda$-terms. Given primitives $0$ and $1$, the language incorporates unitary gates such as Hadamard (\texttt{H}):
\[ \arraycolsep=0.125em\def\arraystretch{1.2}%
\begin{array}{r l@{\hspace{1.5em}}l}
\text{Term}\ t \Coloneqq & x & \text{(variable)} \\
 \mid & \lambda x.t & \text{(abstraction)} \\
 \mid & t_1\ t_2 & \text{(application)} \\
 \mid & c & \text{(constant)} \\
\text{Constant}\ c \Coloneqq & 0 \mid 1 & \text{(bit)} \\
 \mid & \texttt{H} & \text{(Hadamard gate)}
\end{array}
\]

In principle, the calculus may also feature arbitrary other quantum gates, though Hadamard suffices for our argument that the calculus is not realizable.
Also, an equivalent to having primitives $0$ and $1$ would be to replace them with their Church encodings $\lambda x. \lambda y. x$ and $\lambda x. \lambda y. y$.

\paragraph{Semantics}
In the operational semantics of this calculus, the states of the machine are quantum states over terms $t$, such as $\ket{0}$, $\ket{\texttt{H}\ 0}$, and $\frac{1}{\sqrt{2}}(\ket{0} + \ket{1})$.
The work of \citet{vantonder2004} specifies the desired operational semantics of the Hadamard gate to be as follows:
\begin{align*}
\begin{aligned}
\ket{\texttt{H}\ 0} &\mapsto_\beta \textstyle\frac{1}{\sqrt{2}}\ket{0} + \frac{1}{\sqrt{2}}\ket{1} \\
\ket{\texttt{H}\ 1} &\mapsto_\beta \textstyle\frac{1}{\sqrt{2}}\ket{0} - \frac{1}{\sqrt{2}}\ket{1}
\end{aligned}
\end{align*}

That is, the program term $\ket{\texttt{H}\ 0}$ $\beta$-reduces to a superposition of the terms $\ket{0}$ and $\ket{1}$, and the term $\ket{\texttt{H}\ 1}$ reduces to a superposition with opposite phase. This behavior accords with the standard definition of the Hadamard unitary gate (\Cref{sec:background}).

The original work of \citet{vantonder2004} only incompletely specifies how this semantics should evaluate the term $\ket{\texttt{H}\ (\texttt{H}\ 0)}$. However, the standard definition of the Hadamard gate postulates the desired semantics to be the following:
\begin{align}
\begin{split} \label{eqn:double-hadamard}
  & \ket{\texttt{H}\ (\texttt{H}\ 0)} \\
  \mapsto_\beta{} & \textstyle\frac{1}{\sqrt{2}}(\ket{\texttt{H}\ 0} + \ket{\texttt{H}\ 1}) \\
  \mapsto_\beta{} & \textstyle\frac{1}{\sqrt{2}}(\frac{1}{\sqrt{2}}\ket{0} + \frac{1}{\sqrt{2}}\ket{1} + \ket{\texttt{H}\ 1}) \\
  \mapsto_\beta{} & \textstyle\frac{1}{\sqrt{2}}(\frac{1}{\sqrt{2}}\ket{0} \cancel{\textcolor{mygray}{{+} \frac{1}{\sqrt{2}}\ket{1}}} + \frac{1}{\sqrt{2}}\ket{0} \cancel{\textcolor{mygray}{{-} \frac{1}{\sqrt{2}}\ket{1}}}) = \ket{0}
\end{split}
\end{align}

That is, the program term $\ket{\texttt{H}\ (\texttt{H}\ 0)}$ should step to a superposition of the terms $\ket{\texttt{H}\ 0}$ and $\ket{\texttt{H}\ 1}$. Taking more steps to evaluate the inner Hadamard terms should produce a state that thanks to interference (\Cref{sec:background}) is equal to $\ket{0}$.

\paragraph{Linearity}
The key property that enables the \texttt{H} operator in this $\lambda$-calculus to behave as the standard Hadamard gate is that the term transition function $\mapsto$ of the operational semantics acts \emph{linearly} (\Cref{sec:background}) across quantum states:
\begin{align*}
  \textstyle\sum_i \gamma_i \ket{t_i} \mapsto \sum_i \gamma_i \ket{t'_i} \iff \forall i. \ket{t_i} \mapsto \ket{t'_i}
\end{align*}

Linearity postulates that the evaluation of a superposition of program terms produces the corresponding superposition of evaluated program terms. If this property were not true, the final step in the evaluation of $\ket{\texttt{H}\ (\texttt{H}\ 0)}$ would not produce the correct interference that results in $\ket{0}$.

\paragraph{Physical Realizability}
The problem is that the calculus has a transition function that is not injective and hence not unitary. There exist multiple distinct states that transition to the same state:
\begin{align*}
  \ket{0} \neq \ket{\psi} = \textstyle\frac{1}{\sqrt{2}}(\frac{1}{\sqrt{2}}\ket{0} + \frac{1}{\sqrt{2}}\ket{1} + \ket{\texttt{H}\ 1}) \mapsto_\beta& \ket{0} \\
  \ket{0} \mapsto_\beta& \ket{0}
\end{align*}

The first line above is the third reduction step in \Cref{eqn:double-hadamard}. The second line embodies the fact that in $\lambda$-calculus, terms that are not applications are values and step to themselves. More broadly, the non-injectivity of $\beta$-reduction is fundamental to $\lambda$-calculus, as demonstrated by the fact that the distinct terms $\lambda x. x$ and $(\lambda x. x)(\lambda x. x)$ both reduce to the same term $\lambda x. x$.

That the transition function is not injective means that it is not \emph{norm-preserving} (\Cref{sec:background}). There exist physically realizable machine states that transition to unrealizable states:
\begin{align*}
  \textstyle\frac{1}{\sqrt{2}}(\ket{0} - \ket{\psi}) \mapsto_\beta \textstyle\frac{1}{\sqrt{2}}(\ket{0} - \ket{0}) = 0
\end{align*}
in which $\ket{\psi}$ is as above and the output state is the zero vector --- a physically impossible outcome.

\paragraph{Disruptive Entanglement}
The work by \citet{vantonder2004} attempts to resolve the non-injectivity problem via Landauer embedding --- maintaining a history $\ell$ of intermediate terms during evaluation. Just as in \Cref{sec:examples}, the problem of disruptive entanglement arises:
\begin{align*}
  & \ket{\ell, \texttt{H}\ (\texttt{H}\ 0)} \\
  \mapsto_\beta{}& \textstyle\frac{1}{\sqrt{2}}(\ket{\ell, \texttt{H}\ (\texttt{H}\ 0), \texttt{H}\ 0} + \ket{\ell, \texttt{H}\ (\texttt{H}\ 0), \texttt{H}\ 1}) \\
  \mapsto_\beta{}& \textstyle\frac{1}{\sqrt{2}}(\frac{1}{\sqrt{2}}\ket{\ell, \texttt{H}\ (\texttt{H}\ 0), \texttt{H}\ 0, 0} + \frac{1}{\sqrt{2}}\ket{\ell, \texttt{H}\ (\texttt{H}\ 0), \texttt{H}\ 0, 1} + \ket{\ell, \texttt{H}\ (\texttt{H}\ 0), \texttt{H}\ 1}) \\
  \mapsto_\beta{}& \textstyle\frac{1}{\sqrt{2}}(\frac{1}{\sqrt{2}}\ket{\ell, \texttt{H}\ (\texttt{H}\ 0), \texttt{H}\ 0, 0} + \frac{1}{\sqrt{2}}\ket{\ell, \texttt{H}\ (\texttt{H}\ 0), \texttt{H}\ 0, 1} \\
  & \hspace{0.56em}+ \textstyle\frac{1}{\sqrt{2}}\ket{\ell, \texttt{H}\ (\texttt{H}\ 0), \texttt{H}\ 1, 0} - \frac{1}{\sqrt{2}}\ket{\ell, \texttt{H}\ (\texttt{H}\ 0), \texttt{H}\ 1, 1}) \\
  \neq{}& \ket{\psi} \otimes \ket{0}\,\text{for any}\,\ket{\psi}
\end{align*}

Here, entanglement with the history means that no interference occurs between the two terms $\smash{\frac{1}{\sqrt{2}}}\ket{\ell, \texttt{H}\ (\texttt{H}\ 0), \texttt{H}\ 0, 1}$ and $-\smash{\frac{1}{\sqrt{2}}}\ket{\ell, \texttt{H}\ (\texttt{H}\ 0), \texttt{H}\ 1, 1}$. The value produced by evaluation is not a separable qubit $\ket{0}$ but instead entangled with the history. Measuring this value yields 0 with probability only $\frac{1}{2}$, which is inconsistent with the semantics of Hadamard in \Cref{eqn:double-hadamard}.

The work of \citet{vantonder2004} presents attempts to overcome the problem of entanglement by reducing the information that is stored in the history, replacing terms with placeholders. None of these attempts succeed, except a final attempt that fundamentally modifies the language to force all $\lambda$-terms to be classical and prevent terms from existing in superposition.

\paragraph{Implications}
The work of \citet{vantonder2004} concludes, without formal proof, that the above quantum $\lambda$-calculus with program terms in superposition is not physically realizable.

\Cref{thm:no-embedding} formally proves this claim --- because the $\beta$-reduction transition function of $\lambda$-calculus is non-injective, no embedding, including the use of histories as in Landauer embedding, can realize this function on a quantum computer in a way that preserves the correct interference of the Hadamard gate. In turn, the non-injectivity of  $\beta$-reduction is inherent to the nature of the $\lambda$-calculus, as it ensures that both expressions $\lambda x. x$ and $(\lambda x. x)(\lambda x. x)$ reduce to the value $\lambda x. x$.

On the other hand, we show that the more drastic step of forcing all control flow to be classical is not necessary for a quantum computer.
Indeed, the quantum control machine supports control flow in superposition via instructions that have inherently injective semantics. It remains open whether a programming model resembling the $\lambda$-calculus may be used to program this machine.

\section{Incorrect Branch Interference Example} \label{sec:incorrect-example}

\begin{figure}
\[ \arraycolsep=0.125em\def\arraystretch{1.2}%
\begin{array}{ r r l }
& &\!\ket{\texttt{x} \sep 3, \texttt{c} \sep 0, \texttt{pc}_0 \sep 5} \\
  \mapsto{} & \textstyle\frac{1}{\sqrt{2}}(\!&\!\ket{\texttt{x} \sep 3, \texttt{c} \sep 0, \texttt{pc}_0 \sep 9, \texttt{pc}_1 \sep 5} + \ket{\texttt{x} \sep 3, \texttt{c} \sep 1, \texttt{pc}_0 \sep 6, \texttt{pc}_1 \sep 5}) \\
  \mapsto{} & \textstyle\frac{1}{\sqrt{2}}(\!&\!\ket{\texttt{x} \sep 2, \texttt{c} \sep 0, \texttt{pc}_0 \sep 5, \texttt{pc}_1 \sep 9, \ldots} + \ket{\texttt{x} \sep 4, \texttt{c} \sep 1, \texttt{pc}_0 \sep 5, \texttt{pc}_1 \sep 6, \ldots}) \\
  \mapsto{} & \textstyle\frac{1}{2}(\!&\!\ket{\texttt{x} \sep 2, \texttt{c} \sep 0, \texttt{pc}_0 \sep 9, \texttt{pc}_1 \sep 5, \ldots} + \ket{\texttt{x} \sep 2, \texttt{c} \sep 1, \texttt{pc}_0 \sep 6, \texttt{pc}_1 \sep 5, \ldots} \\[-0.2em]
  & +{}\vphantom{\textstyle\frac{1}{\sqrt{2}}} &\!\ket{\texttt{x} \sep 4, \texttt{c} \sep 0, \texttt{pc}_0 \sep 9, \texttt{pc}_1 \sep 5, \ldots} - \ket{\texttt{x} \sep 4, \texttt{c} \sep 1, \texttt{pc}_0 \sep 6, \texttt{pc}_1 \sep 5, \ldots}) \\
  \mapsto{} & \textstyle\frac{1}{2}(\!&\!\ket{\texttt{x} \sep 1, \texttt{c} \sep 0, \texttt{pc}_0 \sep 5, \texttt{pc}_1 \sep 9, \ldots} + \ket{\texttt{x} \sep 3, \texttt{c} \sep 1, \texttt{pc}_0 \sep 5, \texttt{pc}_1 \sep 6, \ldots} \\[-0.2em]
  & +{}\vphantom{\textstyle\frac{1}{\sqrt{2}}} &\!\ket{\texttt{x} \sep 3, \texttt{c} \sep 0, \texttt{pc}_0 \sep 5, \texttt{pc}_1 \sep 9, \ldots} - \ket{\texttt{x} \sep 5, \texttt{c} \sep 1, \texttt{pc}_0 \sep 5, \texttt{pc}_1 \sep 6, \ldots}) \\
  \mapsto{} & \textstyle\frac{1}{2\sqrt{2}}(\!&\!\ket{\texttt{x} \sep 1, \texttt{c} \sep 0, \texttt{pc}_0 \sep 9, \texttt{pc}_1 \sep 5, \texttt{pc}_2 \sep 9, \ldots} + \ket{\texttt{x} \sep 1, \texttt{c} \sep 1, \texttt{pc}_0 \sep 6, \texttt{pc}_1 \sep 5, \texttt{pc}_2 \sep 9, \ldots} \\[-0.3em]
  & +{} &\!\ket{\texttt{x} \sep 3, \texttt{c} \sep 0, \texttt{pc}_0 \sep 9, \texttt{pc}_1 \sep 5, \texttt{pc}_2 \sep 6, \ldots} \textcolor{red}{{-} \ket{\texttt{x} \sep 3, \texttt{c} \sep 1, \texttt{pc}_0 \sep 6, \texttt{pc}_1 \sep 5, \texttt{pc}_2 \sep 6, \ldots}} \\[-0.1em]
  & +{} &\!\ket{\texttt{x} \sep 3, \texttt{c} \sep 0, \texttt{pc}_0 \sep 9, \texttt{pc}_1 \sep 5, \texttt{pc}_2 \sep 9, \ldots} \textcolor{red}{{+} \ket{\texttt{x} \sep 3, \texttt{c} \sep 1, \texttt{pc}_0 \sep 6, \texttt{pc}_1 \sep 5, \texttt{pc}_2 \sep 9, \ldots}} \\[-0.1em]
  & -{}\vphantom{\textstyle\frac{1}{\sqrt{2}}} &\!\ket{\texttt{x} \sep 5, \texttt{c} \sep 0, \texttt{pc}_0 \sep 9, \texttt{pc}_1 \sep 5, \texttt{pc}_2 \sep 6, \ldots} + \ket{\texttt{x} \sep 5, \texttt{c} \sep 1, \texttt{pc}_0 \sep 6, \texttt{pc}_1 \sep 5, \texttt{pc}_2 \sep 6, \ldots}) \\
  \mapsto{} & \textstyle\frac{1}{2\sqrt{2}}(\!&\!\ket{\texttt{x} \sep 0, \texttt{c} \sep 0, \texttt{pc}_0 \sep 14, \texttt{pc}_1 \sep 9, \texttt{pc}_2 \sep 5, \texttt{pc}_3 \sep 9, \ldots} + \ket{\texttt{x} \sep 2, \texttt{c} \sep 1, \texttt{pc}_0 \sep 14, \texttt{pc}_1 \sep 6, \texttt{pc}_2 \sep 5, \texttt{pc}_3 \sep 9, \ldots} \\[-0.3em]
  & +{} &\!\ket{\texttt{x} \sep 2, \texttt{c} \sep 0, \texttt{pc}_0 \sep 14, \texttt{pc}_1 \sep 9, \texttt{pc}_2 \sep 5, \texttt{pc}_3 \sep 6, \ldots} \textcolor{red}{{-} \ket{\texttt{x} \sep 4, \texttt{c} \sep 0, \texttt{pc}_0 \sep 14, \texttt{pc}_1 \sep 6, \texttt{pc}_2 \sep 5, \texttt{pc}_3 \sep 6, \ldots}} \\[-0.1em]
  & +{} &\!\ket{\texttt{x} \sep 2, \texttt{c} \sep 0, \texttt{pc}_0 \sep 14, \texttt{pc}_1 \sep 9, \texttt{pc}_2 \sep 5, \texttt{pc}_3 \sep 9, \ldots} \textcolor{red}{{+} \ket{\texttt{x} \sep 4, \texttt{c} \sep 0, \texttt{pc}_0 \sep 14, \texttt{pc}_1 \sep 6, \texttt{pc}_2 \sep 5, \texttt{pc}_3 \sep 9, \ldots}} \\[-0.1em]
  & -{} &\!\ket{\texttt{x} \sep 4, \texttt{c} \sep 0, \texttt{pc}_0 \sep 14, \texttt{pc}_1 \sep 9, \texttt{pc}_2 \sep 5, \texttt{pc}_3 \sep 6, \ldots} + \ket{\texttt{x} \sep 6, \texttt{c} \sep 1, \texttt{pc}_0 \sep 14, \texttt{pc}_1 \sep 6, \texttt{pc}_2 \sep 5, \texttt{pc}_3 \sep 6, \ldots})
\end{array} \]%
\caption{Incorrect analogue to \Cref{fig:hadamard-walk-trace}, in which a history is kept, and interference does not occur.} \label{fig:hadamard-walk-trace-bad}
\end{figure}

In \Cref{fig:hadamard-walk-trace-bad}, we depict the alternative incorrect execution of \Cref{fig:hadamard-walk} that would result from writing down a history of program counters. Compared to \Cref{fig:hadamard-walk-trace}, the critical difference is that the states highlighted in \textcolor{red}{red} are physically distinct and thus do not interfere to cancel.

\section{Costs of Realization} \label{sec:costs}

In this section, we discuss practical costs to realizing control flow in superposition in terms of hardware and program verification. These costs are inevitable to the concept of a program counter in superposition, whether by the quantum control machine or any other design. By clarifying them here, we aim to help designers of quantum hardware and programming languages understand the feasibility of this programming abstraction in the trajectory of quantum computation.

\paragraph{Hardware Costs}
A practical limitation to realizing a program counter in superposition in hardware is that the control unit introduces substantial qubit and connectivity requirements.

The control unit requires logic to realize instruction fetching, decoding, and execution, which is expensive to implement and error-correct on near-term hardware. These costs are the same as in prior stored-program quantum architectures~\citep{kjaergaard2020,wang2022,meier2024}.
The control unit also conceptually requires connectivity with all data qubits, which is difficult to realize in the near term. These costs are the same as in quantum random-access memory~\citep{giovannetti2008,paler2020,arunachalam2015,matteo2020}.

\paragraph{Program Verification Costs}
It would be ideal to create a sound and complete procedure to verify that any quantum program is synchronized. In the theoretical worst case, this problem has exponential complexity, as is the case for verifying typical nontrivial semantic properties of classical programs with bounded registers.
Nevertheless, for practical instances of programs for the quantum control machine, it is unnecessary to determine the results of an exponential number of executions.

First, verifying a sequence of $n$ binary conditional branches requires checking only that each of the $n$ pairs of syntactic branches align, rather than the $2^n$ possible program paths.
Second, we present in \Cref{sec:iteration-case-study} a reasoning strategy for loops in which the execution time depends on only a classical upper bound, making the program synchronized by construction. Constant-time compilation strategies in secure programming~\citep{barthe2020} similarly apply.

In a more general setting, we expect that given sufficient analysis of program structure, it will be possible to partially decide this condition through techniques developed in classical and quantum program analysis. Five promising approaches are:

\begin{enumerate}
\item Formally verify, only once, the compilation of high-level constructs such as branches and loops, similarly to verified optimizers for quantum circuits~\citep{hietala2021};
\item Generalize the technique of abstracting over quantum variables as in the loop case above, similarly to quantum abstract interpretation~\citep{yu2021};
\item Use the interpretation of synchronization as a $k$-safety property to leverage relational and Cartesian Hoare logics that provide automated verification tools~\citep{sousa2016}, as could be done via quantum relational Hoare logic~\citep{unruh2019};
\item Use the interpretation of synchronization as a bound on the worst-case execution time of a program to leverage analyses of execution time via SMT solvers~\citep{seshia2011};
\item Perform runtime checks during execution on quantum hardware using a small number of measurements~\citep{li2020,twist}.
\end{enumerate}

\section{Other Related Work} \label{sec:related-work}

In this section, we discuss related work that was not addressed in \Cref{sec:implications}, on the topics of quantum $\lambda$-calculi, reversible and oblivious computation, and unbounded-time quantum computation.

\paragraph{Quantum $\lambda$-Calculi}
Control flow in superposition hints at the possibility of a purely functional paradigm of quantum programming.
Though several $\lambda$-calculi for quantum computation have been proposed, their operational semantics are not inherently injective, and they do not permit $\lambda$-terms to exist in superposition.
For example, the quantum $\lambda$-calculus with classical control~\citep{selinger2004} performs $\beta$-reduction classically and does not permit terms to exist in superposition.
Another example is the calculus of \citet{sabry2018}, whose operational semantics of $\lambda$-terms is not injective as both $\lambda x.x$ and $(\lambda x.x)(\lambda x.x)$ reduce to $\lambda x.x$. A possible way toward developing an inherently injective calculus may be to use a type system to reject terms that would reduce in a non-injective manner, as has been pursued by~\citet{arrighi2017,diazcaro2019}.

To our knowledge, the understanding of higher-order quantum programming by categorical approaches~\citep{pagani2014,malherbe2013,rennela2017,hines2011} remains that ``the general connection between physics and higher-order quantum functions is yet unclear''~\citep{renella2018}. Recent developments~\citep{hasuo2017,clairambault19,clairambault20,kornell2021} may be able to shed further light. A related topic is quantum computation with indefinite causal order such as the quantum switch~\citep{chiribella2013}, for which the quantum control machine could serve as a programmable instruction set.

\paragraph{Reversible Computation}
In classical reversible computation, the Landauer embedding~\citep{landauer1961} and uncomputation~\citep{bennett1973} have long been used to simulate arbitrary control flow constructs, including rewriting systems such as combinatory logic~\citep{dipierro2006}.

By contrast, the no-embedding theorem implies that they fail to correctly do so in quantum computation. The reason is entanglement, an inherently quantum phenomenon that does not arise in reversible computation.
The theorem precludes any system of uncomputation that is program-agnostic and always correctly erases the history in the Landauer embedding. It does not preclude program-specialized or sometimes-incorrect systems, and indeed reverse jumps are instances of manual uncomputation~\citep{silq,tower} in quantum programming.

The branch control register was introduced by seminal work on the Pendulum reversible architecture~\citep{axelsen2007,vieri1998,frank1999,axelsen2011}. In reversible computation, the cost of using a history is extra memory or energy consumption~\citep{thomsen2012}. By contrast, in quantum computation, entanglement leads the output to be irrecoverably incorrect.

\paragraph{Oblivious Computation}
Synchronization is related to the property of \emph{obliviousness}~\citep{pippenger1979,goldreich1996} as studied in security and privacy. Obliviousness states that the evolution of the program counter may not depend on the input data, only on its size, and is used in the proof~\citep[Theorem 4.3]{bv1997} that any quantum Turing machine may be transformed into an equivalent one with a classically fixed termination time.

Compared to synchronization, obliviousness is unnecessarily strong --- it constrains all intermediate states of the computation, meaning that the program in \Cref{sec:examples} could not manipulate the program counter based on data if it were required to be oblivious. Synchronization thus lies between the properties of obliviousness and \emph{history independence}~\citep{teague2001,tower}, which fixes the final states, but not running times, of a computation.

\paragraph{Unbounded-Time Quantum Computation}
A restriction of the quantum control machine is that the running time of a computation must be a classically determined value. This restriction is fundamentally implied by synchronization, which states that there must exist a point in time at which all branches of the machine state superposition possess the same control state.

Evidence suggests that such a restriction is inevitable to quantum computation, and that it is not possible to physically realize a quantum loop that does not depend on any classical bound:

\begin{itemize}
\item \citet{myers1997,kieu1998,linden1998} show that in a quantum Turing machine, when different branches of the state superposition of a computation take different amounts of time to terminate, no general scheme is known for determining whether all branches have halted that does not collapse the superposition.
\item \citet{ying2010} indicate that ``If a quantum loop is allowed to be embedded into another quantum loop, then \ldots{} the body of the latter loop is not a unitary operator but a super-operator in general'', making it not realizable as a circuit of unitary logic gates.
\item \citet{andresmartinez2022} show that a loop-like operation is possible via weak measurements of the system. The issue is that if the strength parameter used by the scheme is too strong or weak, the result is incorrect, and in turn determining its optimal value requires access to problem structure analogous to a classical iteration bound.
\end{itemize}

As a result, the quantum control machine has been designed to accord with the standard assumption in the theory of quantum computation~\citep[Section 3.5]{bv1997} that the termination time of a quantum computation is classically fixed.

\end{document}